\documentclass[10pt,british]{article}
\pdfoutput=1

\usepackage[hmarginratio=1:1,top=19mm,left=27mm]{geometry}

\usepackage[T1]{fontenc}
\usepackage[utf8]{inputenc}
\usepackage{microtype}

\usepackage{babel}
\usepackage{csquotes}
\usepackage[all]{foreign}
\defasforeign[mutatismutandis]{mutatis mutandis}
\defasforeign[Mutatismutandis]{Mutatis mutandis}
\defasforeign[criterium]{criterium}
\defnotforeign[wrt]{{\xperiodafter{w.r.t}}}

\usepackage{amsthm}
\usepackage{amsmath}
\usepackage{amsfonts,amssymb,bm,stmaryrd}
\usepackage{mathtools}
\usepackage{esvect}

\allowdisplaybreaks

\usepackage{thmtools}

\usepackage{float}
\usepackage{caption}

\usepackage{enumitem}

\usepackage{tikz}
\usetikzlibrary{calc,arrows,positioning,intersections}

\usepackage{hyperref}
\usepackage{nameref}

\usepackage[capitalise,noabbrev,nameinlink]{cleveref}

\hypersetup{
	hidelinks,
	final,
}

\usepackage{autonum}

\usepackage[
	backend=bibtex,
	style=numeric-comp,
	maxbibnames=99, %
	giveninits=true,
	sorting=nyt,
	doi=false,
	isbn=false,
	url=false
	]{biblatex}
\theoremstyle{plain}
\newtheorem{theorem}{Theorem}[section]
\newtheorem{corollary}[theorem]{Corollary}

\newtheorem{lemma}[theorem]{Lemma}
\theoremstyle{definition}
\newtheorem{definition}{Definition}[section]

\newtheorem*{notation}{Notation}
\theoremstyle{remark}
\newtheorem{remark}{Remark}[section]
\newtheorem{example}[remark]{Example}

\makeatletter
\newcommand{\Superimpose}[2]{%
	{\ooalign{$#1\@firstoftwo#2$\cr\hfil$#1\@secondoftwo#2$\hfil\cr}}}
\makeatother

\makeatletter
\newcommand{\footnoteref}[1]{%
	\protected@xdef\@thefnmark{\ref{#1}}\@footnotemark%
}
\makeatother

\makeatletter
\newcommand{\customlabel}[2]{%
	\protected@write \@auxout {}{\string \newlabel {#1}{{#2}{\thepage}{#2}{#1}{}} }%
	\hypertarget{#1}{#2}
}
\makeatother

\newcommand{\oversetlabel}[3]{
	\phantomsection
	\overset{({\customlabel{#2}{#1}})}{#3}
}

\newcounter{proofstep}
\newcommand*{\stepref}[1]{}
\newcommand*{\steplabel}[2]{}
\newcommand*{\setupstepsequence}[2][0]{%
	\renewcommand*{\stepref}[1]{\eqref{#2-##1}}%
	\renewcommand*{\steplabel}[2][\theproofstep]{%
		\stepcounter{proofstep}%
		\oversetlabel{\textit{\roman{proofstep}}}{#2-##1}{##2}%
	}%
	\setcounter{proofstep}{#1}%
}

\DeclareMathOperator{\img}{img}
\DeclareMathOperator{\car}{car}
\DeclareMathOperator{\bis}{bis}
\DeclareMathOperator{\supp}{supp}

\newcommand{\ltrue}{\texttt{t\!t}}
\newcommand{\lfalse}{\texttt{f\!f}}

\DeclareBoldMathCommand{\bflangle}{\langle}
\DeclareBoldMathCommand{\bfrangle}{\rangle}

\newcommand{\cat}[1]{\textnormal{\textsf{#1}}\xspace}
\newcommand{\Set}{\cat{Set}}
\newcommand{\Sys}[1]{{\cat{Sys}(#1)}}

\newcommand{\ff}[1]{\mathcal{F}_{\mspace{-1mu}#1}}
\newcommand{\fp}[1]{\mathcal{P}_{\mspace{-1mu}#1}}
\newcommand{\fpf}{\fp{\mspace{-2mu}f}}

\newcommand{\defeq}{\triangleq}
\newcommand{\defiff}{\stackrel{\triangle}{\iff}}
\newcommand{\dotrel}[1]{\mathrel{\dot{#1}}}

\newcommand{\reduceTo}{\preccurlyeq}
\newcommand{\reduceFrom}{\succcurlyeq}
\newcommand{\reduceEq}{\approxeq}
\newcommand{\freduceTo}{\dotrel{\reduceTo}}

\newcommand{\freduceEq}{\dotrel{\reduceEq}}

\newcommand\restr[2]{{\left.\kern-\nulldelimiterspace#1\vphantom{\big|}\right|_{#2}}}
\newcommand\corestr[2]{{\left.\kern-\nulldelimiterspace#1\vphantom{\big|}\right|^{#2}}}

\title{
	Reductions for Transition Systems at Work:\\%
	Deriving a Logical Characterization of Quantitative Bisimulation
}
\author{
	{Marino Miculan}%
	\thanks{Partially supported by MIUR project 2010LHT4KM (\emph{CINA}).}\\
	\small University of Udine, Italy\\
	\small \href{mailto:marino.miculan@uniud.it}{\tt marino.miculan@uniud.it}
\and
	{Marco Peressotti}%
	\thanks{Supported by the CRC project, grant no.~DFF–4005-00304 from the Danish Council for Independent Research, and by the Open Data Framework project at the University of Southern Denmark.}\\
	\small University of Southern Denmark\\
	\small \href{mailto:peressotti@imada.sdu.dk}{\tt peressotti@imada.sdu.dk}
}

\date{}

\begin{document}

\maketitle

\begin{abstract}
	Weighted labelled transition systems (WLTSs) are an established (meta-)model aiming to provide general results and tools for a wide range of systems such as non-deterministic, stochastic, and probabilistic systems. In order to encompass processes combining several quantitative aspects, extensions of the WLTS framework have been further proposed,  \emph{state-to-function transition systems} (FuTSs) and \emph{uniform labelled transition systems} (ULTraSs) being two prominent examples.
	In this paper we show that this hierarchy of meta-models collapses when studied under the lens of bisimulation-coherent encodings.
	Taking advantage of these reductions, we derive a \emph{fully abstract} Hennessy-Milner-style logic for FuTSs, \ie,  which characterizes quantitative bisimilarity, from a fully-abstract logic for WLTSs.
\end{abstract}


\section{Introduction}\label{sec:intro}

\emph{Weighted labelled transition systems} (WLTSs) \cite{ks:ic2013-sos} are a meta-model for systems with quantitative aspects: transitions $P\xrightarrow{a,w}Q$ are labelled with \emph{weights} $w$, taken from a given monoidal weight structure. Many computational aspects can be captured just by changing the underlying weight structure: weights can model probabilities, resource costs, stochastic rates, \etc.; as such, WLTSs are a generalisation of labelled transition systems (LTSs) \cite{milner:cc}, probabilistic systems (PLTSs) \cite{gsb90:ic}, stochastic systems \cite{hillston:pepabook}, among others. Definitions and results developed in this setting instantiate to existing models, thus recovering known results and discovering new ones. In particular, the notion of \emph{weighted bisimulation} \cite{ks:ic2013-sos} in WLTSs coincides with (strong) bisimulation for all the aforementioned models.

In the wake of these encouraging results, other meta-models have been proposed aiming to cover an even wider range of computational models and concepts.
\emph{Uniform labelled transition systems} (ULTraSs) \cite{denicola13:ultras} are systems whose transitions have the form $P\xrightarrow{a}\phi$, where $\phi$ is a \emph{weight function} assigning weights to states; hence, ULTraSs can be seen both as a non-deterministic extension of WLTSs and as a generalisation of Segala's probabilistic systems \cite{sl:njc95} (NPLTSs).
In \cite{mp:qapl14,mp:tcs2016} a (coalgebraically derived) notion of bisimulation for ULTraSs is presented and shown to precisely capture bisimulations for weighted and Segala systems.
\emph{Function-to-state transition systems} (FuTSs) were introduced in \cite{denicola13:ustoc} as a generalisation of the above, of IMCs \cite{hermanns:imcbook}, and other models. Later, \cite{latella:lmcs2015} defined a (coalgebraically derived) notion of (strong) bisimulation for FuTSs which instantiates to known bisimulations for all the aforementioned models.

Given all these meta-models, it is natural to wonder about their \emph{expressiveness}.
We should consider not only the class of systems these frameworks can represent,  but also \emph{whether} these representations are faithful with respect to the properties we are interested in.  
Intuitively, a meta-model $\cat{M}$ is \emph{subsumed by $\cat{M}'$ according to a property $P$} if any system $S$ which is an instance of $\cat{M}$ with the property $P$, is also an instance of $\cat{M}'$ preserving $P$.

In this work we study these meta-models according to their ability to correctly express \emph{strong bisimulation}.  In this context, a meta-model $\cat{M}$ is \emph{subsumed} by $\cat{M}'$ if any system $S$ which is an instance of $\cat{M}$, is also an instance of $\cat{M}'$ preserving and reflecting strong bisimulations.
\begin{figure}
	\begin{center}
		\begin{tikzpicture}[
				font=\small,auto,
				xscale=1.2, yscale=.8,
				baseline=(current bounding box.center),
			]
			\node (futs)   at (.5,3) {\cat{FuTS}};
			\node (ultras) at (.5,2) {\cat{ULTraS}};
			\node (wlts)   at ( 1,1) {\cat{WLTS}};
			\node[gray] (nplts)  at ( 0,1) {\cat{NPLTS}};
			\node[gray] (lts)    at ( 0,0) {\cat{LTS}};
			\node[gray] (plts)   at ( 1,0) {\cat{PLTS}};

			\node[left of=lts,gray] {\dots};
			\node[right of=plts,gray] {\dots};

			\draw (futs) -- (ultras);
			\draw (ultras) -- (wlts);
			\draw[gray] (ultras) -- (nplts);
			\draw[gray] (nplts) -- (plts);
			\draw[gray] (nplts) -- (lts);
			\draw[gray] (wlts) -- (plts);
			\draw[gray] (wlts) -- (lts);
		\end{tikzpicture}
		\caption{The bisimulation-driven hierarchy of function-to-state transition systems.}
		\label{fig:hierarchy}
	\end{center}
\end{figure}
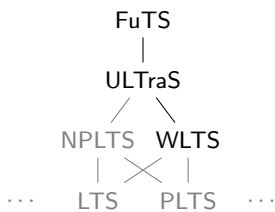
Previous work  \cite{ks:ic2013-sos,denicola13:ultras,mp:qapl14,mp:tcs2016,latella:lmcs2015} has shown that, according to this order, each of the meta-models mentioned above subsumes the previous ones, thus forming the hierarchy shown in \cref{fig:hierarchy}.
Still, an important question is open: 
\begin{quotation}
Is any of these meta-models \emph{strictly more expressive} than others?
\end{quotation}
In this work we address this question, proving that this is not the case: the black part of the hierarchy collapses!

In this venture, we introduce a new notion of  \emph{reduction between classes of systems}. We use these reductions to formally define the expressiveness order between system classes with respect to (strong) bisimulation, but this notion is more general and can be used to study any class of state-based transition systems, since all the constructions and results are developed abstracting from the ``type'' of computation under scrutiny. 
Besides the classification interest, reductions pave the way for porting existing and new results between categories of transition systems. In this paper, we use reductions to define new Hennessy-Milner-style modal logics for transition systems.  An important property for these logics is to be \emph{fully abstract} \ie to characterise bisimilarity: ``two states are logically equivalent (\ie, satisfy precisely the same set of formulae) if and only if they are bisimilar''.  We illustrate how reductions can be used to infer this relevant property; then, as an application, we introduce \emph{finite-conjunction logic for FuTSs} and prove that it is fully abstract via a reduction.

\paragraph{Remark} This work builds on the conference paper \cite{mp:ictcs2016}.  Main notions and results discussed in  \cref{sec:transition-systems,sec:equivalence-extensions,sec:bisimulations,sec:reductions} can be found also in \loccit, but in \cref{sec:futs-reductions} we propose a reduction of FuTSs to WLTSs which is different, and simpler, with respect to that in \cite{mp:ictcs2016}.  Moreover, in this work we have added new results concerning logical characterizations of bisimulations (\cref{sec:reductions-logic,sec:futs-logic-fully-abstract}).  

\paragraph{Synopsis}
\Cref{sec:transition-systems} recalls an abstract and uniform account of transition systems on discrete state spaces, akin to \cite{rutten:universal}.
\Cref{sec:equivalence-extensions} presents a general construction for extending equivalence relations over sets of states to sets of behaviours.
Building on this relational extension, \cref{sec:bisimulations} provides a characterisation of (strong) bisimulations in a modular fashion.
The notion of reduction is introduced in \cref{sec:reductions}, along with general reductions.
In \cref{sec:futs-reductions} we provide a reduction from the category of FuTSs to the category of WLTSs together with intermediate reductions for special cases of FuTSs such as ULTraS. 
In \cref{sec:fully-abstract-logic} we recall modal logics and in \cref{sec:reductions-logic} we extend the notion of reduction from categories of systems to modal logics and illustrate their use for proving full abstraction.
In \cref{sec:futs-logic-fully-abstract} we introduce a logic for FuTSs and prove that it characterises bisimulation via a reduction to a fully abstract logic for WLTSs.
Final remarks are in \cref{sec:conclusions}.

\section{Discrete transition systems}
\label{sec:transition-systems}
For an alphabet $A$ and set of states $X$, the function space $X^A$ is understood as the set of all possible behaviours characterising deterministic input over $A$. In this context, a transition system exposing this computational behaviour is precisely described by a function $\alpha\colon X \to X^A$ mapping each state $x \in X$ to some element in $X^A$. For a function $f\colon X \to Y$ and $\phi \in X^A$, the assignment $\phi \mapsto f \circ \phi$ defines a function $(f)^A\colon X^A \to Y^A$ that extends the action of $f$ from  state spaces $X$ and $Y$ to behaviours defined over them in a coherent way. A function $f\colon X \to Y$ between the state spaces of systems (say, $\alpha\colon X \to X^A$ and $\beta\colon Y \to Y^A$) preserves and reflect their structure whenever $f^A \circ \alpha = \beta \circ f$.
Intuitively this means that if $\alpha$ transits from $x$ to $x'$ when inputs $a$ then the same input makes $\beta$ transit from $f(x)$ to $f(x')$ and \viceversa.
Since they preserve and reflect the transition structure of systems, these functions are called \emph{homomorphisms} (which are functional bisimulations, \cf \cite[Thm.~2.5]{rutten:universal}).

All the structures and observations described in the above example stem from a single information: the ``type'' of the behaviour under scrutiny. This is well understood as an \emph{endofunctor} over the category of state spaces \cite{rutten:universal} --- in this context, the category of sets and functions.

Non-deterministic transitions are captured by the powerset endofunctor $\mathcal{P}$ mapping each set $X$ its powerset $\mathcal{P}X$ and function $f$ to its inverse image $\mathcal{P}f$ \ie the function given by the assignment $Z \mapsto \{ f(z) \mid z \in Z\}$.
Since subsets are functions weighting elements over the monoid $\mathbb{B} = (\{\ltrue,\lfalse\},\lor,\lfalse)$, the above readily extends to quantitative aspects (such as probability distributions, stochastic rates, delays, etc.) by simply considering other a non-trivial abelian monoids\footnote{An abelian monoid is a set $M$ equipped with an associative and commutative binary operation $+$ and a unit $0$ for $+$; such structure is called trivial when $M$ is a singleton.}  \cite{ks:ic2013-sos,mp:tcs2016,latella:qapl2015}. This yields the endofunctor $\ff{M}$ which assigns
\begin{itemize}
	\item
		to each set $X$ the set $\{ \phi\colon X \to M \mid \supp(\phi) \text{ is finite}\}$ of finitely supported weight functions (the support of $\phi$ is the set $\{ x \mid \phi(x) \neq 0\}$);
	\item
		to each function $f\colon X\to Y$ the map $(\ff{M}f)(\phi) = \lambda y \in Y.\sum_{x:f(x)=y} \phi(x)$. (summation is well defined because $\phi$ is finitely supported).
\end{itemize}
\begin{notation}
We will often denote elements of $\ff{M}X$ using the formal sum notation: for $\phi \in \ff{M}X$ we write $\sum_x \phi(x) \cdot x$ or, given $\supp(\phi) = \{x_1,\ldots,x_n\}$, simply $\sum_{i=1,\ldots,n} \phi(x_i)\cdot x_i$. For instance, $\ff{M} (f)(\phi)$ is formulated as $\sum \phi(x) \cdot f(x)$.
\end{notation}
As discussed \eg in \cite[Sec.~2]{mp:tcs2016}, it is indeed possible to consider supports of greater cardinalities given that the definition of $\ff{M}$ is restricted those abelian monoids equipped with sums for families of the desired cardinality. When instantiated on $\mathbb{B}$, the monoid of boolean values under disjunction, the above is equivalent to the finite powerset $\fpf$. Likewise, natural numbers under addition yields finite multisets.
Probabilistic computations are a special case of the above where weight functions are distributions (\cf \cite{ks:ic2013-sos}) and are captured by the endofunctor $\mathcal{D}$ given on each set $X$ as $\mathcal{D}X = \{\phi \in \ff{[0,\infty)}X \mid \sum\phi(x) = 1\}$ and on each function $f$ as $\ff{[0,\infty)}f$.

From this perspective, $\mathcal{D}$ can be thought as a sort of ``subtype'' of $\ff{[0,\infty)}$. This situation is formalised by means of (component-wise) injective natural transformations (herein \emph{injective transformations}).
Composition and products of natural transformations are component-wise and the class of injective ones is closed under such operations. 
In general, for an injective transformation $\mu$ and a n endofunctor $T$, $\mu_T$ is again injective but $T\mu$ may not be so. The latter is injective given that $T$ \emph{preserves injective maps} \ie $Tf$ is injective whenever $f$ is injective. 
This mild assumption is met by all examples considered in this paper and is preserved by endofunctor composition and products---these endofunctors are also known as \emph{(generalised) Kripke polynomial endofunctors} \cite{bonsangue:lics2009}.
\begin{lemma}
	\label{thm:inj-preservation}
	Let $T$ be any endofunctor described by the grammar:
	\begin{equation}\textstyle
	 	T \Coloneqq Id \mid A \mid T + T \mid \prod_{i \in I} T_i \mid \ff{M}
	\end{equation}
	Injective functions are preserved by $T$.
\end{lemma}

\begin{example}
	\label{ex:cast-wlts-ultras}
	The endofunctor $\mathcal{P}\ff{M}$ models the alternation of non-deterministic steps with quantitative aspects captured by $(M,+,0)$.
	There is an injective transformation $\eta\colon Id \to \mathcal{P}$ whose components are given by the mapping $x \mapsto \{x\}$ and hence, by composition, $\eta_{\ff{M}} \colon \ff{M} \to \mathcal{P}\ff{M}$ is an injective transformation.
	\qed
\end{example}

\begin{definition}
	For an endofunctor $T$ over $\Set$, a transition system of type $T$ (\emph{$T$-system}) is a pair $(X,\alpha)$ where $X$ is the set of states (\emph{carrier}) and $\alpha\colon X \to TX$ is the \emph{transition map}.
	For $(X, \alpha)$ and $(Y, \beta)$ $T$-systems, a \emph{$T$-homomorphism} from the former to the latter is a function $f\colon X \to Y$ such that $Tf \circ \alpha = f \circ \beta$.
\end{definition}

Since system homomorphism composition is defined in terms of composition of the underlying functions on carriers it is immediate to check that the operation is associative and has identities. Therefore, any class of systems together with their homomorphisms defines a category.

We adopt the following notational conventions.
\begin{notation}
	A transition system $(X,\alpha)$ is referred by its transition map only; in this case its carrier is written $\car(\alpha)$.
	Homomorphisms are denoted by their underlying function.
	Categories of systems are written using sans serif font with $\Sys{T}$ being the category of all $T$-systems and $T$-homomorphisms and $\cat{C}|_{T}$ its subcategory of systems in the category $\cat{C}$.
\end{notation}

\begin{example}[LTSs]
For a set $A$ of labels, labelled transition systems are $(\mathcal{P}-)^A$-systems, and image finite LTSs $(\fpf-)^A$-systems \cite{rutten:universal}.
Hereafter let $\cat{LTS}$ denote the category of all image-finite labelled transition systems and let $\cat{LTS}(A) \defeq \Sys{(\fpf-)^A}$ be its subcategory of systems labelled over $A$.
\qed
\end{example}

\begin{example}[WLTSs]
For a set of labels $A$ and an abelian monoid $M$, weighted labelled transition systems are characterised by the endofunctor $(\ff{M}-)^A$ \cite{ks:ic2013-sos} and hence form the category $\cat{WLTS}(A,M) \defeq \Sys{(\ff{M}-)^A}$ \ie the $(A,M)$-indexed component of $\cat{WLTS}$, the category of all WLTSs.
When the monoid $\mathbb{B}$ of boolean values under disjunction is considered, $\cat{WLTS}(A,\mathbb{B})$ is $\cat{LTS}(A)$.
\qed
\end{example}

\begin{example}[ULTraSs]
We adopt the presentation of ULTraSs given in \cite{mp:qapl14,mp:tcs2016}.
For a set of labels $A$ and an abelian monoid $M$, uniform labelled transition systems are characterised by the endofunctor $(\mathcal{P}\ff{M}-)^A$; image finite ULTraSs by $(\fpf\ff{M}-)^A$.
We denote by $\cat{ULTraS}$ the category of all image-finite ULTraSs
and by $\cat{ULTraS}(A,M)$ its subcategory of systems with labels in $A$ and weights in $M$. WLTSs can be cast to ULTraSs by means of the injective transformation $(\eta_{\ff{M}})^A$ described in \cref{ex:cast-wlts-ultras}. 
As shown in \cite{mp:qapl14}, this natural transformation is actually embeds the semantics of WLTSs (in the sense of \cref{sec:reductions}) into a special class of ULTraSs called called in \cite{denicola13:ultras} \emph{functional}.
\qed
\end{example}

\begin{example}[FuTSs]
FuTSs are $T$-systems for $T$ generated by the grammar
\begin{equation}
	T \Coloneqq (S-)^A \mid T \times (S-)^A
	\qquad
	S \Coloneqq \ff{M} \mid \ff{M} \circ S
\end{equation}
where $A$ and $M$ range over (non-empty) sets of labels and (non-trivial) abelian monoids, respectively. Any such endofunctor is equivalently described by:
\begin{equation}\textstyle
	(\ff{\vv{M}}f)^{\vv{A}}
	\defeq
	\prod_{i=0}^{n} (\ff{\vv{M}_i}f)^{A_i}
	\qquad\text{and}\qquad
	(\ff{\vv{M}_i}f)^{A_i}
	\defeq
	(\ff{M_{i,0}}\dots \ff{M_{i,m_i}}f)^{A_i}
\end{equation}
for $\vv{A} = \langle A_0,\dots,A_n\rangle$ a sequence of non-empty sets, $\vv{M}_i = \langle M_{i,0},\dots,M_{i,l_i}\rangle$ a sequence of non-trivial abelian monoids, and $\vv{M} = \langle \vv*{M}{0},\dots,\vv*{M}{n}\rangle$.
(Up to minor notational variations, this characterisation can be found in \cite{latella:qapl2015,mp:tcs2016}.)
For any $\vv{A}$ and $\vv{M}$ as above define $\cat{FuTS}\big(\vv{A},\vv{M}\big)$ as $\Sys{(\ff{\vv{M}}-)^{\vv{A}}}$. Clearly, $\cat{FuTS}(\langle A\rangle,\langle M\rangle)$ and $\cat{FuTS}(\langle A\rangle,\langle B,M\rangle)$ coincide with $\cat{WLTS}{(A,M)}$ and $\cat{ULTraS}(A,M)$, respectively. Then, \cat{LTS}, \cat{WLTS}, and \cat{ULTraS} are subcategories of $\cat{FuTS}$, the category of all FuTSs.

For $\vv{M} = \langle\langle M_{0,0},\dots,M_{0,l_0}\rangle,\dots,\langle M_{n,0} \dots M_{n,l_n}\rangle\rangle$ as above, recall from \cite{latella:qapl2015} that a FuTS over $\vv{M}$ is called:
\begin{itemize}
	\item
		\emph{nested} whenever $n=0$,
	\item
		\emph{combined} whenever $m_i = 0$ for each $i \in \{0,\dots,n\}$, and
	\item
		\emph{simple} whenever it is both combined and nested.
\end{itemize}
The categories of nested, combined, and simple FuTSs are denoted as $\cat{N-FuTS}$, $\cat{C-FuTS}$, and $\cat{S-FuTS}$, respectively. In particular, $\cat{S-FuTS}$ and $\cat{WLTS}$ coincide.
\qed
\end{example}

\section{Equivalence extensions}
\label{sec:equivalence-extensions}
Several definitions of bisimulation found in literature use (more or less explicitly) some sort of extension of equivalence relations from state spaces to behaviours over these spaces. For instance, in \cite{sl:njc95} two probability distributions are \mbox{considered} equivalent with respect to an equivalence relation $R$ on their domain if they assign the same probability to any equivalence class induced by $R$:
\begin{equation}
	\phi \equiv_R \psi \defiff \forall C \in X/R\ \left(\sum_{x \in C}\phi(x) = \sum_{x \in C}\psi(x)\right)\text{.}
\end{equation}
This section defines equivalence extensions for arbitrary endofunctors (over \Set) and studies how constructs such as composition or products reflect on these extensions, providing some degree of modularity.

\begin{definition}
	For an equivalence relation $R$ on $X$ its \emph{$T$-extension} is the equivalence relation $R^T$ on $TX$:
	\begin{equation}
		\phi \mathrel{R^T} \psi \defiff (T\kappa)(\phi) = (T\kappa)(\psi)
	\end{equation}
	where $\kappa\colon X \to X/R$ is the canonical projection to the quotient induced by $R$.
\end{definition}

As an example, let us consider the endofunctor $(-)^A$ describing deterministic inputs on $A$: the resulting extension for an equivalence relation $R$ relates functions mapping the same inputs to states related by $R$.
\begin{equation}
	\phi \mathrel{R}^{(-)^A} \psi
	\iff
	\kappa\circ\phi = \kappa \circ \psi
	\iff
	\forall a \in A\ (\phi(a) \mathrel{R} \psi(a))
	\text{.}
\end{equation}
Extensions for $\mathcal{P}$ are precisely ``subset closure'' of relations (\cf \cite{mp:tcs2016}) and relate all and only those subsets for which the given relation is a correspondence. Formally:
\begin{equation}
	Y \mathrel{R}^{\mathcal{P}} Z
	\iff
	\{\kappa(y) \mid y \in Y\} = \{\kappa(z) \mid z \in Z\}
	\iff
	(\forall y \in Y \exists z \in Z (y \mathrel{R} z))
	\land
	(\forall z \in Z \exists y \in Y (y \mathrel{R} z))
\end{equation}
Extension for $\ff{M}$ are generalise the subset closure to multisets and relate only weight functions assigning the same cumulative weight to each equivalence class induced by $R$:
\begin{equation}
	\phi \mathrel{R}^{\ff{M}} \psi
	\iff 
	\sum \phi(x) \cdot \kappa(x) = \sum \psi(x) \cdot \kappa(x)
	\iff 
	\forall C \in X/R\left(\sum_{x\in C} \phi(x) = \sum_{x\in C} \psi(x)\right)
	\text{.}
\end{equation}
In particular, ${R^\mathcal{D}}$ is precisely Segala's equivalence $\equiv_R$ \cite{sl:njc95}.

Consider extensions for the endofunctor $(\mathcal{P}-)^A$ describing LTSs:
\begin{align}
	\phi \mathrel{R}^{(\mathcal{P}-)^A} \psi
	& \iff
	\sum\left(\sum \phi(a)(x) \cdot \kappa(x)\right)\cdot a = 
	\sum\left(\sum \psi(a)(x) \cdot \kappa(x)\right)\cdot a
	\\ & \iff
	\forall a \in A \left(
	(\forall y \in \phi(a)\, \exists z \in \psi(a)\, (y \mathrel{R} z))
	\,\land\,
	(\forall z \in \psi(a)\, \exists y \in \phi(a)\, (y \mathrel{R} z))
	\right)
	\text{.}
\end{align}
Clearly, $\mathrel{R}^{\mathcal{P}(-)^A}$ can be equivalently written as
\begin{equation}
	\phi \mathrel{R}^{\mathcal{P}(-)^A} \psi \iff \forall a \in A (\phi(a) \mathrel{R}^{\mathcal{P}} \psi(a))
\end{equation}
which suggests some degree of modularity in the definition of extensions to composite endofunctors. This kind of reformulations is not possible since for arbitrary endofunctors $T$ ans $S$, it holds only that
\begin{equation}
	\phi \mathrel{\left(R^S\right)^T} \psi \Longrightarrow \phi \mathrel{R^{T\circ S}} \psi\text{.}
\end{equation}
The converse implication holds whenever $T$ preserves injections.
\begin{lemma}
	\label{thm:extension-composition}
	For $T$ and $S$ endofunctors, 
	\begin{itemize}
		\item
			$\left(R^S\right)^T \subseteq R^{T\circ S}$ and, 
		\item
			given $T$ preserves injective functions, $\left(R^S\right)^T \supseteq R^{T\circ S}$.
	\end{itemize}
\end{lemma}

\begin{proof}
	Let $\kappa^S\colon SX \to SX/R^S$ be the canonical projection to the quotient induced by the equivalence relation $R^S$. Since, by definition,
	$\kappa^S(\rho) = \kappa^S(\theta)$ implies $(S\kappa)(\rho) = (S\kappa)(\theta)$ there is a (unique) function $q^S\colon SX/R^S \to S(X/R)$ such that $S\kappa = q^S \circ \kappa^S$.
	From  $TS\kappa = Tq^S \circ T\kappa^S$ and the definition of $\mathrel{R^{TS}}$ and $\left(R^S\right)^T$, it follows that:
	\begin{equation}
		\phi \mathrel{\left(R^S\right)^T} \psi
		\implies
		(TS\kappa)(\phi) = (TS\kappa)(\psi) \implies
		\phi \mathrel{R^{TS}} \psi
	\end{equation}
	proving first part of the thesis.
	Since $\rho \mathrel{R^S} \theta \iff \kappa^S(\rho) = \kappa^S(\theta)$ we conclude that $q^S$ is an injection and, by hypothesis, $Tq^S$ is an injection too. Therefore:
	\begin{equation}
		\phi \mathrel{R^{TS}} \psi
		\implies
		(TS\kappa)(\phi) = (TS\kappa)(\psi)
		\implies
		(T\kappa^S)(\phi) = (T\kappa^S)(\psi)
		\implies
		\phi \mathrel{\left(R^S\right)^T} \psi
	\end{equation}
	completing the proof.
\end{proof}

Endofunctors modelling inputs, such as $(-)^A$ and $(\fpf-)^A$, can be seen as products (in these cases as powers) of endofunctors indexed over the input space $A$. As suggested by the above examples, for product endofunctors it holds that:
\begin{equation}
	\phi \mathrel{R}^{(\prod T_i)} \psi \iff \forall i \in I(\pi_{i}(\phi) \mathrel{R}^{T_i} \pi_i(\psi))
\end{equation}
where $\pi_i\colon \prod T_i X \to T_i X$ is the projection on the $i$-th component of the product.

\begin{lemma}
	\label{thm:extension-product}
	For $I \neq \emptyset$ and $\{T_i\}_{i \in I}$,
	$\mathrel{R}^{\left(\prod_{i \in I} T_i\right)} \cong \prod_{i \in I}\mathrel{R}^{T_i}$.
\end{lemma}

\begin{proof}
	Write $T$ for $\prod_{i \in I} T_i$ and
	recall that $\left(\prod_{i \in I} T_i\right)X$ is  $\prod_{i \in I} T_iX$.
	Then:
	\begin{equation}
		\phi \mathrel{R}^{T} \psi
		\textstyle
		\iff
		(\prod T_i\kappa)(\phi) = (\prod T_i\kappa)(\psi)
		\iff
		\prod (T_i\kappa)(\phi_i) = \prod (T_i\kappa)(\psi_i)
		\iff
		\phi \mathrel{\prod R^{T_i}} \psi
	\end{equation}
	where $\kappa\colon X \to X/R$ is the canonical projection to the quotient induced by $R$
	and $\pi_i\colon \prod_{i \in I} T_iX \to T_iX$ is the $i$-th projection.
\end{proof}

FuTSs offer an instance of the above result: the endofunctor $(\ff{\vv{M}}-)^{\vv{A}}$ modelling FuTSs over $\vv{M} = \langle\vv{M}_{0},\dots,\vv{M}_{n}\rangle$ and $\vv{A} = \langle A_0;\dots;A_n\rangle$ is a product indexed over $\{(i,a) \mid i \leq n \land a \in A_i\}$. Thus, the extension $R^{(\ff{\vv{M}}-)^{\vv{A}}}$ is described by:
\begin{equation}
	\phi \mathrel{R}^{(\ff{\vv{M}}-)^{\vv{A}}} \psi
	\iff
	\forall i \leq n \forall a \in A_i
	(\phi_i(a) \mathrel{R}^{\ff{\vv{M}_i}} \psi_i(a))
	\text{.}
\end{equation}

For an equivalence relation $R$ define its restriction to $X$ as the equivalence relation $R|_X \defeq R \cap (X \times X)$.
Both $(R|_X)^T$ and $R^T|_{TX}$ are equivalence relations over the set of $T$-behaviours for $X$ and, in general, the former is finer than the latter, unless $T$ preserves injections---in such case, the two coincide.

\begin{lemma}
	\label{thm:extension-and-restriction}
	For $R$ and equivalence relation on $Y$ and $X \subseteq Y$,
	\begin{itemize}
		\item
			$(R|_X)^T \subseteq R^T|_{TX}$, and, 
		\item
			provided $T$ preserves injective functions, $(R|_X)^T \supseteq R^T|_{TX}$.
	\end{itemize}
\end{lemma}

\begin{proof}
	Let $\kappa\colon Y \to Y/R$ and $\kappa'\colon X \to X/{R|_X}$ be the canonical projections induced by $R$ and $R|_X$, respectively. Since the latter is given by restriction of the former to $X \subseteq Y$, there is a unique and injective map $q\colon X/{R|_X} \to Y/R$ such that $\kappa = q \circ \kappa$. The first part of the thesis follows by:
	\begin{equation}
		\phi \mathrel{(R|_X)}^{T} \psi
		\implies
		(T\kappa')(\phi) = (T\kappa')(\psi)
		\implies
		(T\kappa)(\phi) = (T\kappa)(\psi)
		\implies
		\phi \mathrel{R}^T \psi
	\end{equation}
	since $T\kappa = Tq \circ T\kappa$.
	On the other hand, by hypothesis on $T$, $Tq$ is injective and hence
	\begin{equation}
		\phi \mathrel{R^{T}|_{TX}} \psi
		\implies
		(T\kappa)(\phi) = (T\kappa)(\psi)
		\implies
		(T\kappa')(\phi) = (T\kappa')(\psi)
		\implies
		\phi \mathrel{(R|_X)}^T \psi
		\tag*{\qedhere}
	\end{equation}
\end{proof}

Intuitively, this result allows us to encode multiple steps sharing the same computational aspects as single steps at the expense of bigger state spaces. In fact, it follows that $(R|_X)^{T^{n+1}} = R^T|_{T^{n}X}$, assuming $T$ preserves injections.

\begin{lemma}
	\label{thm:injective-nat-extension}
	Let $\mu\colon T \to S$ be an injective natural transformation.
	For $R$ an equivalence relation on $X$,
	\begin{equation}
		\phi \mathrel{R}^T \psi \iff \mu_{X}(\phi) \mathrel{R}^S \mu_{X}(\psi)
		\text{.}
	\end{equation}
\end{lemma}

\begin{proof}
	\setupstepsequence{step:injective-nat-extension}
	It holds that
	\begin{align}
		\phi \mathrel{R}^T \psi
		& \iff
		T\kappa(\phi) = T\kappa(\psi)
		\steplabel{\iff}
		(\mu_X \circ T\kappa)(\phi) = (\mu_X \circ T\kappa)(\psi)
		\steplabel{\iff}
		(S\kappa\circ \mu_X)(\phi) = (S\kappa\circ \mu_X)(\psi)
		\\&\iff
		\mu_{X}(\phi) \mathrel{R}^S \mu_{X}(\psi)
	\end{align}
	where \stepref{1} and \stepref{2} follow by $\mu_X$ being injective and by $\mu$ being a natural transformation, respectively.
\end{proof}

\section{Bisimulations}\label{sec:bisimulations}
In this section we give a general definition of bisimulation based on the notion of equivalence relation extension introduced above.
This approach is somehow modular, as the definition reflects the structure of the endofunctors characterising systems under scrutiny. This allows to extend results developed in \cref{sec:equivalence-extensions} to bisimulation and, in \cref{sec:reductions}, to reductions.

\begin{definition}
	\label{def:bisim-via-extensions}
	An equivalence relation $R$ is a \emph{strong $T$-bisimulation} (herein, bisimulation) for a $T$-system $\alpha$ if and only if:
	\begin{equation}
		x \mathrel{R} x' \implies \alpha(x) \mathrel{R^T} \alpha(x')
		\text{.}
	\end{equation}
	States $x$ and $x'$ of $\alpha$ are called \emph{bisimilar} (written $x \sim x'$) whenever there exists a bisimulation $R$ for $\alpha$ such that $x \mathrel{R} x'$.
	The set of all bisimulations for the system $\alpha$ is denoted by $\bis(\alpha)$.
\end{definition}

\noindent
The notion of bisimulations as per \cref{def:bisim-via-extensions} coincides with Aczel-Mendler's notion of \emph{precongruence} \cite{am89:final}.

\begin{definition}
	\label{def:precongruence}
	An equivalence relation $R$ on $X$ is a (Aczel-Mendler) precongruence for $\alpha\colon X \to TX$ if, and only if, for any two functions $f,f' \colon X \to Y$ such that $x \mathrel{R} x' \implies f(x) = f'(x')$ it holds that \begin{equation}
		x \mathrel{R} x' \implies (Tf \circ \alpha)(x) = (Tf' \circ \alpha)(x')
		\text{.}
	\end{equation}
\end{definition}

\begin{theorem}
	\label{thm:bisim-precongruence}
	For $\alpha$ a $T$-system, every strong $T$-bisimulation for $\alpha$ is an (Aczel-Mendler) precongruence and \viceversa.
\end{theorem}

\begin{proof}
	\setupstepsequence{step:bisim-precongruence-1}
	Assume $R$ is a bisimulation for $\alpha\colon X \to TX$. For $f,f'\colon X \to Y$ such that $x \mathrel{R} x' \implies f(x) = f'(x')$ we have that:
	\begin{equation}
		x \mathrel{R} x
		\implies
		\alpha(x) \mathrel{R^T} \alpha(x')
		\iff
		(T\kappa \circ \alpha)(x) = (T\kappa \circ \alpha)(x')
		\steplabel{\implies}
		(Tf \circ \alpha)(x) = (Tf' \circ \alpha)(x')
	\end{equation}
	where \stepref{1} follows by noting that, since $\kappa\colon X \to X/R$ is a canonical projection and 
	\begin{equation}
		x \mathrel{R} x' \implies f(x) = f'(x')
		\text{,}
	\end{equation}
	there is (a unique) $q\colon X/R \to Y$ such that $f = q\circ \kappa = f'$.

	\setupstepsequence{step:bisim-precongruence-2}
	Assume $R$ is a precongruence for $\alpha$, we have that:
	\begin{equation}
		x \mathrel{R} x
		\steplabel{\iff}
		\kappa(x) = \kappa(x')
		\steplabel{\implies}
		(T\kappa \circ \alpha)(x) = (T\kappa \circ \alpha)(x')
		\steplabel{\iff}
		\alpha(x) \mathrel{R^T} \alpha(x')
	\end{equation}
	where \stepref{1} follows from the definition of $\kappa\colon X \to X/R$,
	\stepref{2} from the assumption that $R$ is a precongruence, and \stepref{3} from the definition of $R^T$.
\end{proof}

Bisimulations for systems considered in this paper are known to be \emph{kernel bisimulations} (\cf \cite{rutten:universal,ks:ic2013-sos,mp:tcs2016,latella:qapl2015}) \ie kernels of functions carrying homomorphisms from systems under scrutiny  \cite{staton11}. These can be intuitively thought as defining \emph{refinement systems} over the equivalence classes they induce.

\begin{definition}
	\label{def:kernel}
	A relation $R$ on $X$ is a kernel bisimulation for $\alpha\colon X \to TX$ if, and only if, there is $\beta\colon Y \to TY$ and $f\colon \alpha \to \beta$ such that $R$ is the kernel of the map underlying $f$.
\end{definition}

In general, \cref{def:bisim-via-extensions} is stricter than \cref{def:kernel} but the two coincide for endofunctors preserving (enough) injections---\eg any example from this paper.

\begin{corollary}
	\label{thm:bisim-kernel}
	For $\alpha\colon X \to TX$, the following are true:
	\begin{itemize}
		\item
			A bisimulation for $\alpha$ is a kernel bisimulation for $\alpha$.
		\item
			If $T$ preserves injections then, a kernel bisimulation for $\alpha$ is a bisimulation for $\alpha$.
	\end{itemize}
\end{corollary}

\begin{proof}
	By \cref{thm:bisim-precongruence} and \cite[Thm.~4.1]{staton11}.
\end{proof}

This result relates \cref{def:bisim-via-extensions} with the notions of bisimulations found in the literature for the models considered in this paper.
In particular, it follows from \cref{thm:bisim-kernel,thm:inj-preservation} that \cref{def:bisim-via-extensions} captures strong bisimulation for LTSs \cite{milner:cc}, WLTSs \cite{ks:ic2013-sos}, NPLTS \cite{sl:njc95}, ULTraSs \cite{mp:tcs2016}, and FuTSs \cite{latella:qapl2015}
since these are all known to be instances of kernel bisimulation (see \loccit).

\begin{lemma}
	\label{thm:bisim-product}
	For $T = \prod_{i \in I}T_i$ and $\alpha$ a $T$-system, $\bis(\alpha) = \bigcap_{i \in I} \bis(\pi_i \circ \alpha)$.
\end{lemma}

\begin{proof}
	It follows from \cref{thm:reduction-product} that:
	\begin{equation}
		\alpha(x) \mathrel{R^T} \alpha(x') \iff \alpha(x) \mathrel{\Big(\prod_{i\in I}R^{T_i}\Big)} \alpha(x')
	\end{equation} 
	from which we conclude that:
	\begin{equation}
		\alpha(x) \mathrel{R^T} \alpha(x') \iff \forall i \in I\ (\pi_i \alpha)(x) \mathrel{R^{T}_i} (\pi_i \alpha)(x')
		\text{.}
		\tag*{\qedhere}
	\end{equation}
\end{proof}

A special but well known instance of \cref{thm:bisim-product} is given by definitions of bisimulations found in the literature for LTSs, WLTSs and in general FuTSs. In fact, all these bisimulation contain a universal quantification over the set of labels. For instance, a $R$ is a bisimulation for an LTS $\alpha\colon X \to (\mathcal{P}X)^A$ iff:
\begin{align}
	x \mathrel{R} x' & \implies 
	\forall a \in A \left(\alpha(x)(a) \mathrel{R^{\mathcal{P}}} \alpha(x')(a)\right) 
	\\ & \iff 
	\forall a \in A \left(
	(\forall y \in \phi(a)\, \exists z \in \psi(a)\, (y \mathrel{R} z))
	\land
	(\forall z \in \psi(a)\, \exists y \in \phi(a)\, (y \mathrel{R} z))
	\right)	
\end{align}
that is, if, and only if, $R$ is the intersection of an $A$-indexed family composed by a bisimulation for each transition system $\alpha_{a}\colon X \to \mathcal{P}X$ projection of $\alpha$ on $a \in A$.

\begin{lemma}
	\label{thm:bisim-flattening}
	For $n \in \mathbb{N}$ and $\alpha$ a  $T^{n+1}$-system, there is $T$-system $\underline{\alpha}$ such that the following implications hold:
	\begin{itemize}
		\item
			$R \in \bis(\alpha) \implies \exists R' \in \bis(\underline{\alpha})(R = R'|_{\car(\alpha)})$,
		\item
			$R \in \bis(\underline{\alpha}) \implies R|_{\car(\alpha)} \in \bis(\alpha)$.
	\end{itemize}
\end{lemma}

\begin{proof}
	\setupstepsequence{step:bisim-flattening-1}

	For $\alpha\colon X \to T^{n+1}X$ define $\underline{\alpha}\colon \underline{X} \to T(\underline{X})$ as:
	\begin{equation}
		\underline{X} \defeq \coprod_{i=0}^{n} T^iX 
		\qquad
		\underline{\alpha} \defeq [T\iota_n \alpha_0, T\iota_0\alpha_1 \dots, T\iota_{n-2}\alpha_{n-1}]
	\end{equation}
	where $\iota_i\colon T^{i}X \to \underline{X}$ is the $i$-th coproduct injection, $\alpha_0\colon X \to T(T^n X)$ is $\alpha$, and $\alpha_{i+1}\colon T^{i+1}X \to T(T^i X)$ is given by the identity for $T^{i+1}X$. In particular, 
	\begin{equation}
		\underline{\alpha}(\underline{x}) = 
		\begin{cases}
			\alpha(\underline{x}) & \text{if }\underline{x} \in X\\
			\underline{x} & \text{otherwise}
		\end{cases}
	\end{equation}
	for any $\underline{x} \in \underline{X}$.
 	If $R \in \bis(\overline{\alpha})$, then:
	\begin{equation}
		x \mathrel{R|_X} x'
		\implies
		x \mathrel{R} x'
		\steplabel{\implies}
		\underline{\alpha}(x) \mathrel{R}^T \underline{\alpha}(x')
		\steplabel{\iff}
		\alpha(x) \mathrel{R^T|_{T^{n+1}X}} \alpha(x')
		\steplabel{\iff}
		\alpha(x) \mathrel{(R|_X)}^{T^{n+1}} \alpha(x')
	\end{equation}
	where \stepref{1} follows by $R \in \bis(\underline{\alpha})$; \stepref{2} follows by noting that $\underline{\alpha}$ acts as $\alpha$ on $X$ and hence both $\underline{\alpha}(x)=\alpha(x)$ and $\underline{\alpha}(y) = \alpha(y)$ are elements of $T^{n+1}X$; \stepref{3} follows by inductively applying \cref{thm:extension-composition,thm:extension-and-restriction}. Therefore, $R|_X \in \bis(\alpha)$.
	\setupstepsequence{step:bisim-flattening-2}
	Assume $R \in \bis(\alpha)$ and define $\underline{R} = \coprod_{i=0}^n R^{T^i}$. By construction of $\underline{R}$, $\underline{x} \mathrel{\underline{R}} \underline{x'}$ implies that $\underline{x}, \underline{x'} \in T^{i}X$ for some $i \in \{0,\dots,n\}$ meaning that the proof can be carried out by cases on each $R^{T^i}$ composing $\underline{R}$.
	Assume $\underline{x}, \underline{x'} \in T^{0}X = X$, then:
	\begin{equation}
		\underline{x} \mathrel{\underline{R}} \underline{x'}
		\iff
		\underline{x} \mathrel{R} \underline{x'}
		\steplabel{\implies}
		\alpha(\underline{x}) \mathrel{R}^{T^{n+1}} \alpha(\underline{x'})
		\steplabel{\iff}
		\underline{\alpha}(\underline{x}) \mathrel{(R^{T^{n}})^T} \underline{\alpha}(\underline{x'})
		\iff
		\underline{\alpha}(\underline{x}) \mathrel{\underline{R}^T} \underline{\alpha}(\underline{x'})
	\end{equation}
	where \stepref{1} and \stepref{2} follow by $R \in \bis(\alpha)$ and \cref{thm:extension-composition}, respectively.
	\setupstepsequence{step:bisim-flattening-3}%
	Assume $\underline{x}, \underline{x'} \in T^{i+1}X$, we have that:
	\begin{equation}
		\underline{x} \mathrel{\underline{R}} \underline{x'}
		\iff
		\underline{x} \mathrel{R^{T^{i+1}}} \underline{x'}
		\steplabel{\iff}
		\underline{x} \mathrel{(R^{T^{i}})^T} \underline{x'}
		\steplabel{\iff}
		\underline{\alpha}(\underline{x}) \mathrel{(R^{T^{i}})^T} \underline{\alpha}(\underline{x'})
		\iff 
		\underline{\alpha}(\underline{x}) \mathrel{\underline{R}}^T \underline{\alpha}(\underline{x'})
	\end{equation}
	where \stepref{1} and \stepref{2} follow from \cref{thm:extension-composition} and definition of $\underline{\alpha}$ on $T^{i+1}X$. Thus, $\underline{R} \in \bis(\overline{\alpha})$ and $\underline{R}|_X = R$.
\end{proof}
\Cref{thm:bisim-flattening} and its proof provide us with an encoding from systems whose steps are composed by multiple substeps to systems of substeps while preserving and reflecting their semantics in term of bisimulations. The trade-off of the encoding is a bigger state space due to the explicit account of intermediate steps.

\begin{lemma}
	\label{thm:bisim-natural}
	For $\mu\colon T \to S$ injective and $\alpha$ a $T$-system, the following statements are equivalent:
	\begin{itemize}
	 	\item $R$ is a bisimulation for $\alpha \colon X \to TX$;
	 	\item $R$ is a bisimulation for $\mu_X \circ \alpha \colon X \to SX$.
	\end{itemize}
\end{lemma}

\begin{proof}
	Recall that, $R \in \bis(\alpha)$ if, and only if, $x \mathrel{R} y \implies \alpha(x) \mathrel{R}^T \alpha(y)$ and that $R \in \bis(\mu_X \circ \alpha)$ iff:
	\begin{equation}
		x \mathrel{R} y \implies (\mu_X\circ\alpha)(x) \mathrel{R}^S (\mu_X\circ\alpha)(y)
		\text{.}
	\end{equation}
	Therefore to prove the statement it suffices to show that 
	\begin{equation}
		\phi \mathrel{R}^T \psi \iff \mu_X(\phi) \mathrel{R}^S \mu_X(\psi)
		\text{.}
	\end{equation} 
	We conclude by \cref{thm:injective-nat-extension}.
\end{proof}

By applying the \cref{thm:bisim-natural} to \cref{ex:cast-wlts-ultras} we conclude that that bisimulations for ULTraSs coincide with bisimulations for WLTSs when these are seen as functional ULTraS as shown in \cite{mp:qapl14,mp:tcs2016}.

\section{Reductions}
\label{sec:reductions}
In this section we formalize the intuition that a behaviour ``type'' is (at least) as expressive as another whenever systems and homomorphisms of the latter can be ``encoded'' as systems and homomorphisms of the former, provided that their semantically relevant structures are preserved and reflected.

\begin{definition}
	\label{def:system-reduction}
	For $\alpha$ and $\beta$ (of possibly different types), a \emph{system reduction} $\sigma\colon \alpha \to \beta$ is given by
	 \begin{enumerate}
	 	\item
	 		a function $\sigma^{c}\colon \car(\alpha) \to \car(\beta)$
			and
		\item
			a correspondence\footnote{A left-total and right-total binary relation.} $\sigma^{b}\subseteq \bis(\alpha) \times \bis(\beta)$ 
	\end{enumerate}
	such that $\sigma^{c}$ carries a relation homomorphism for any pair of bisimulations in $\sigma^{b}$, \ie:
	\begin{equation}
		\label{eq:reduction-bisim-condition}
		R \mathrel{\sigma^{b}} R' \implies
		(x \mathrel{R} x' \iff \sigma^{c}(x) \mathrel{R'} \sigma^{c}(x'))\text{.}
	\end{equation}
	A system reduction $\sigma\colon \alpha \to \beta$ is called \emph{full} if $\sigma^{c}\colon \car(\alpha) \to \car(\beta)$ is surjective.
\end{definition}

The notion of fullness is of relevance since it identifies reductions that use the target state space in its entirety and hence that do not introduce any auxiliary state. 
A consequence of condition \eqref{eq:reduction-bisim-condition} from \cref{def:system-reduction} is that correspondences forming system reductions are always left-unique: this is indeed stronger than requiring preservation of bisimilarity since it entails that any bisimulation for $\alpha$ can be recovered by restricting some bisimulation for $\beta$ to the image of $\car(\alpha)$ in $\car(\beta)$ through the map $\sigma^{c}$. 

\begin{remark}
	\label{rm:relaxed-versions}
	Condition \eqref{eq:reduction-bisim-condition} can be relaxed in two ways:
	\begin{enumerate}
		\item $R \mathrel{\sigma^{b}} R' \implies (x \mathrel{R} x' \implies \sigma^{c}(x) \mathrel{R'} \sigma^{c}(x'))$,
		\item $R \mathrel{\sigma^{b}} R' \implies (x \mathrel{R} x' \impliedby \sigma^{c}(x) \mathrel{R'} \sigma^{c}(x'))$.
	\end{enumerate}
	The first condition requires every bisimulation for $\alpha$ to be contained in some bisimulation for $\beta$ whereas second requires every bisimulation for $\alpha$ to contain some bisimulation for $\beta$. Hence the two can be thought as \emph{completeness} and \emph{soundness} conditions for the reduction $\sigma$, respectively.
\qed
\end{remark}

System reductions can be extended to whole categories of systems by equipping functors with them and ensuring they respect the structure of homomorphisms. Formally:
\begin{definition}
	\label{def:reduction}
	For $\cat{C}$ and $\cat{D}$ categories of system, a \emph{reduction} $\sigma$ from $\cat{C}$ to $\cat{D}$, written $\sigma\colon \cat{C} \to \cat{D}$, is functor equipped with a collection of system reductions
	\[
		\{\sigma_{\alpha}\colon \alpha \to \sigma(\alpha)\}_{\alpha \in \cat{C}}
	\]
	coherent with homomorphisms in the sense that:
	\begin{equation}
		\label{eq:reduction-c-homomorphism}
		\sigma_{\beta}^{c} \circ f = \sigma(f) \circ \sigma_{\alpha}^{c}
	\end{equation}
	for any $f\colon \alpha \to \beta$ in $\cat{C}$.
	A reduction $\sigma\colon \cat{C} \to \cat{D}$ is called \emph{full} if, and only if, every system reduction $\sigma_\alpha$ is full. A category $\cat{C}$ is said to \emph{reduce} (resp.~fully reduce) to $\cat{D}$, if there is a reduction (resp.~a full reduction) going from $\cat{C}$ to $\cat{D}$.
\end{definition}

\noindent Although \cref{def:reduction} and \cite[Def.~7]{mp:ictcs2016} are slightly different in their presentation, the two are equivalent.

\begin{notation}
	For categories $\cat{C}$ and $\cat{D}$ we write $\cat{C} \reduceTo \cat{D}$ if  $\cat{C}$ reduces to $\cat{D}$, $\cat{C} \reduceEq \cat{D}$ if $\cat{C} \reduceTo \cat{D}$ and $\cat{C} \reduceFrom \cat{D}$, $\cat{C} \freduceTo \cat{D}$ and $\cat{C} \freduceEq \cat{D}$ if the reductions involved are full.
\end{notation}

Reductions can be easily composed at the level of their defining assignments.
In particular, for reductions $\sigma\colon \cat{C} \to \cat{D}$ and $\tau\colon \cat{D} \to \cat{E}$, their composite reduction $\tau \circ \sigma\colon \cat{C} \to \cat{E}$ is a mapping that assigns to each system $\alpha$ the system $(\tau \circ \sigma)(\alpha)$ and the reduction given by $(\tau \circ \sigma)_\alpha^{c} \defeq \tau_{\sigma(\alpha)}^{c} \circ \sigma_\alpha^{c}$ and $(\tau \circ \sigma)_\alpha^{b} \defeq \tau_{\sigma(\alpha)}^{b} \circ \sigma_\alpha^{b}$; and to each $f\colon \alpha \to \alpha'$ the homomorphism $(\tau \circ \sigma)(f)$.
Reduction composition is associative and admits identities which are given on every $\cat{C}$ as the identity assignments for systems and homomorphisms. Any reduction restricts to a reduction from a subcategory of its domain and extends to a reduction to a super-category of its codomain. Moreover, fullness is preserved by the above operations. Every inclusion functor identifies a full reduction.

For products, reductions can be given component-wise by suitable families of reductions that are ``well-behaved'' on homomorphisms. Formally:

\begin{definition}
	A family of reductions $\{\sigma_i\colon \cat{C}_i \to \cat{D}_i\}_{i \in I}$ is called \emph{coherent} if, and only if, the following conditions hold for any $i,j \in I$:
	\begin{enumerate}
	\item
		if a function $f$ extends to $f_i \in \cat{C}_i$ then there is $f_j \in \cat{C}_j$ such that $f$ extends to $f_j$;
	\item
		$\sigma_i(f_i)$ and $\sigma_j(f_j)$ share their underlying function whenever $f_i$ and $f_j$ do.
	\end{enumerate}
\end{definition}

\begin{theorem}
	\label{thm:reduction-product}
	A coherent family of (full) reductions $\{\sigma_i\colon \Sys{T_i} \to \Sys{S_i} \}_{i \in I}$ defines a (full) reduction  $\sigma\colon \Sys{\prod_{i \in I} T_i} \to \Sys{\prod_{i \in I} S_i}$.
\end{theorem}

\begin{proof}
	Assume $\{\sigma_i\}_{i \in I}$ as above. For $\alpha\in \Sys{\prod T_i}$ let $\alpha_i = \pi_i \circ \alpha$ and define
	\begin{equation}
		\sigma(\alpha) \defeq \langle \dots, \sigma_i(\alpha_i), \dots \rangle \qquad\textstyle
		\sigma_{\alpha}^{c} \defeq \sigma_{i,\alpha_i}^{c} \qquad
		\sigma_{\alpha}^{b} \defeq \bigcap_{i \in I}\sigma_{i,\alpha_i}^{b}
	\end{equation}
	The assignment extends to all systems in $\Sys{\prod_{i \in I} T_i}$ and is well-defined by coherency and \cref{thm:bisim-product} since:
	\begin{equation}
		\forall i \in I(\sigma_{\alpha}^{c} = \sigma_{i,\alpha_i}^{c})
		\quad\text{ and }\quad
		R \mathrel{\sigma_{\alpha}^{b}} R' \iff \forall i \in I (R \mathrel{\sigma_{i,\alpha_i}^{b}} R')
		\text{.}
	\end{equation}
	For any $i \in I$, $f\colon \alpha \to \beta$ defines an homomorphism $f_i\colon \alpha_i \to \beta_i$ in $\Sys{T_i}$ sharing its underlying function. Define $\sigma(f)$ as the homomorphism arising from the function underlying $\sigma_i(f_i)$. By coherency, the mapping is well-defined and satisfies all the necessary conditions since all $\sigma_i$ are reductions.
\end{proof}

Correspondences for bisimulations presented in \cref{thm:bisim-natural,thm:bisim-flattening} extend to reductions: injective transformations define full reductions and homogeneous systems reduce to systems for the base endofunctor, as formalised below.

\begin{theorem}
	\label{thm:reduction-natural}
	For $\mu\colon T \to S$ an injective transformation, there is a full reduction from $\Sys{T}$ to $\Sys{S}$.
\end{theorem}
\begin{proof}
	For each each $\alpha$ and each $f\colon \alpha \to \beta$ define
	\begin{equation}
		\hat{\mu}(\alpha) \defeq \mu_{\car(\alpha)} \circ \alpha
		\qquad
		\hat{\mu}_{\alpha}^{c} \defeq id_{\car(\alpha)}
		\qquad
		\hat{\mu}_{\alpha}^{b} \defeq id_{\bis(\alpha)}
		\qquad
		\hat{\mu}(f) \defeq f
		\text{.}
	\end{equation}
	By \cref{thm:bisim-natural} and definition unfolding, the above defines a full reduction $\hat{\mu}\colon \Sys{T} \to \Sys{S}$.
\end{proof}

This theorem allows us to formalise the hierarchy shown in \cref{sec:intro}. For instance, the transformation described in \cref{ex:cast-wlts-ultras} defines a full reduction from WLTSs to ULTraSs. Probabilistic systems are covered by the transformation induced by the inclusion $\mathcal{D}X \subseteq \ff{[0,\infty)}X$ whereas the remaining cases are trivial.

\begin{theorem}
	\label{thm:reduction-flattening}
	If $T$ preserves injections then $\Sys{T^{n+1}}$ reduces to $\Sys{T}$.
\end{theorem}

\begin{proof}
	Recall from \cref{thm:bisim-flattening} that any $T^{n+1}$-system $\alpha\colon X \to T^{n+1}X$ reduces to a $T$-system $\underline{\alpha}\colon \underline{X} \to T\underline{X}$ and let $\iota_0\colon X \to \underline{X}$ denote the injection into the coproduct $\underline{X}$ forming the reduction given in the proof of \cref{thm:bisim-flattening}.
	Define $\sigma\colon \Sys{T^{n+1}} \to \Sys{T}$ as the reduction given on each system $\alpha$ in $\Sys{T^{n+1}}$ as:
	\begin{equation}
		\sigma(\alpha) \defeq \underline{\alpha}
		\quad
		\sigma_{\alpha}^{c} \defeq \iota_0
		\quad
		\sigma_{\alpha}^{b} \defeq \{(R,\underline{R}) \mid R = \underline{R}|_X,\ R \in \bis(\alpha),\ \underline{R} \in \bis(\underline{\alpha})\}
	\end{equation}
	and on each homomorphism $f\colon \alpha \to \beta$ in $\Sys{T^{n+1}}$ as $\sigma(f) \defeq \coprod_{i=0}^{n} T^{i}f$.
	By \cref{thm:bisim-flattening}, $\sigma_{\alpha}^{b}$ is a correspondence and by construction $\sigma$ respects homomorphism composition and identities. Thus, $\sigma$ is a reduction from $\Sys{T^{n+1}}$ to $\Sys{T}$.
\end{proof}

\Cref{thm:reduction-flattening} shows that the encoding proposed in the proof of \cref{thm:bisim-flattening} yields a reduction that intuitively renders explicit intermediate stages of the computation by means of auxiliary states.

\section{Application: reducing FuTSs to WLTSs}
\label{sec:futs-reductions}
In this section we apply the theory presented in the previous sections to prove that (categories of) FuTSs reduce to (categories of) simple FuTSs, \ie WLTSs. In \cite{mp:ictcs2016} we derived a suitable reduction in stages reflecting the endofunctors structure and hence some secondary results regarding the subclasses of nested and combined FuTSs. In this work we propose an alternative reduction which is more suited for the constructions we introduce in \cref{sec:futs-logic-fully-abstract}.

\paragraph{Unlabelled FuTSs}
We call a FuTS \emph{unlabelled} whenever all sets of labels in the sequence $\vv{A}$ defining its type are singletons. In the sequel we adopt the convention of prefixing subcategories of unlabelled FuTSs with $\cat{U}$ and \eg write $\cat{U-FuTS}$ for the category of all unlabelled FuTSs. We claim  that the category of FuTSs fully reduces to its subcategory of unlabelled FuTSs (see \cref{thm:futs-to-u-futs} below). To this end, we present a reduction that ``encodes'' the information of labels in the weighting structure.

For an abelian monoid $(M,+,0)$ and a non-empty set $A$, the function space $M^A$ carries an abelian monoid structure given by an $A$-indexed product of monoids. In particular, its sum and zero are defined as follows:
\[
	\sum_{a \in A} m_a \cdot a + \sum_{a \in A} m'_a \cdot a = 
	\sum_{a \in A} (m_a + m'_a) \cdot a
	\qquad
	\sum_{a \in A} 0 \cdot a
	\text{.}
\]
There is a natural isomorphism:
\begin{equation}
	(\ff{M}-)^A \cong \ff{M^A}
\end{equation}
whose components are given on each set $X$ by the assignments:
\begin{equation}
	\phi \mapsto \sum_{x \in X} \lambda a \in A.\phi(a)(x) \cdot x
\qquad
	\psi \mapsto \lambda a \in A. \sum_{x \in X} \psi(x)(a) \cdot x
\end{equation}
going from $(\ff{M} X)^A$ to $\ff{M^A}X$ and back. It follows from the universal property of products that these indeed exhibit an isomorphism and are natural in the set $X$. 

\begin{lemma}
	\label{thm:futs-to-u-futs}
	The category of FuTSs fully reduces to that of unlabelled FuTSs, hence:
	\[
		\cat{FuTS} \freduceEq \cat{U-FuTS}
		\text{.}
	\]
\end{lemma}

\begin{proof}
	For any $M$ and $A$, the two directions of the natural isomorphism $(\ff{M}-)^A \cong \ff{M^A}$ are componentwise injective natural transformations. Then, the thesis follows from \cref{thm:reduction-natural}.
\end{proof}

\noindent
We denote the full reduction described above as $\sigma_{u}\colon\cat{FuTS} \to \cat{U-FuTS}$.

A special case of unlabelled FuTSs are unlabelled simple ones \ie ``unlabelled WLTSs''. We write $\cat{WTS}$ for their category. It follows from \cref{thm:futs-to-u-futs} that the category of WLTSs fully reduces to that of weighted transition systems, hence that $\cat{WLTS} \freduceEq \cat{WTS}$.

\paragraph{Tabular FuTSs}
We call a FuTS \emph{tabular} whenever the collection $\vv{M}$ defining its type is actually a table \ie whenever $\vv{M} = \langle\langle M_{0,0},\dots,M_{0,l_0}\rangle,\dots,\langle M_{n,0} \dots M_{n,l_n}\rangle\rangle$ is such that $l_i = l_{i'}$ for all $i,i' \in \{0,\dots,n\}$. In this setting we say that $\vv{M}$ is of size $n \times l$. In the sequel we adopt the convention of prefixing subcategories of tabular FuTSs with $\cat{T}$ and \eg write $\cat{T-FuTS}$ for the category of all tabular FuTSs. We claim  that the category of FuTSs fully reduces to its subcategory of tabular FuTSs (see \cref{thm:futs-to-t-futs} below). Intuitively, this reduction introduces some ``functional'' or ``deterministic'' in correspondence of monoids added to fill the gaps in $\vv{M}$ and turn it into a table. Any non-trivial monoid structure can be used for this purpose once a non-zero weight is selected since functional steps are essentially weight functions with supports that are singletons like Dirac's delta function. 
This approach is essentially the same that reduces WLTSs as functional ULTraSs  via the transformation described in \cref{ex:cast-wlts-ultras}. 

Drawing from this observation and to in order to simplify the exposition, we adopt the convention of putting collections of monoids in tabular form by padding them on the left (lower indexes). Fix a non-trivial abelian monoid $\mathbb{P}$. For $\vv{M} = \langle\langle M_{0,0},\dots,M_{0,l_0}\rangle,\dots,\langle M_{n,0} \dots M_{n,l_n}\rangle\rangle$, define $\vv{[M]}$ as the collection of size $n\times l$ where $l = \max\{l_0,\dots,l_n\}$ and
\begin{equation}
[M]_{i,j} \defeq \begin{cases}
	M_{i, j-l+l_i} & \text{if } j > l - l_i\\
	\mathbb{P} & \text{otherwise}
\end{cases}
\text{.}
\end{equation}

\begin{remark}
	We did not specify a concrete choice for $\mathbb{P}$ in the definition of $\vv{[M]}$ since this would impose unnecessary constraints; for instance, in \cref{sec:futs-logic-fully-abstract} we need to restrict to a certain class of monoids that does not contain $\mathcal{B}$.	
\end{remark}

\begin{lemma}
	\label{thm:futs-to-t-futs}
	The category of FuTSs fully reduces to that of tabular FuTSs, hence:
	\[
		\cat{FuTS} \freduceEq \cat{T-FuTS}
		\text{.}
	\]
\end{lemma}

\begin{proof}
	Fix a non-trivial monoid $\mathbb{P}$ and a non-zero element $\mathtt{p}$.
	The assignment  $x \to \mathtt{p} \cdot x$
	induces a componentwise injective natural transformation $\eta\colon Id \to \ff{\mathbb{P}}$.
	By basic properties of natural transformations, this extends to the componentwise injective natural transformation
	\[
		\prod_{i=0}^{n}(\underbrace{\eta \circ \dots \circ \eta}_{l-l_i\text{ times}} \circ Id_{\ff{M_{i,0}}} \circ \dots \circ Id_{\ff{M_{i,l_i}}})^{A_i}
	\]
	($\circ$ denotes horizontal composition in $\cat{Cat}$)
	whose type is $(\ff{\vv{M}}-)^{\vv{A}} \to (\ff{\vv{[M]}}-)^{\vv{A}}$. Then, the thesis follows from \cref{thm:reduction-natural}.
\end{proof}

\noindent
We denote the full reduction described above as $\sigma_{t}\colon\cat{FuTS} \to \cat{T-FuTS}$.

\paragraph{Homogeneous FuTSs} 
We call a FuTS \emph{homogeneous} whenever its weights are all drawn from the same monoid \ie whenever all monoids in the sequence $\vv{M}$ defining its type are the same. In the sequel we adopt the convention of prefixing subcategories of homogeneous FuTSs with $\cat{H}$ and \eg write $\cat{H-FuTS}$ for the category of all homogeneous FuTSs. We claim that the category of FuTSs fully reduces to its subcategory of homogeneous FuTSs (see \cref{thm:futs-to-u-futs} below). Actually, this reduction is an instance of a more general ``relabelling'' reduction where weights are replaced accordingly to selected monoid homomorphisms\footnote{%
	A function $f$ between the carriers of monoids $(M,m,e)$ and $(M',m',e')$ carries a monoid homomorphisms provided it respects multiplications
	(\ie $f(m(x,y)) = m'(f(x),f(y))$ for all $x,y \in M$) and identities ($f(e) = e'$).
}.

For $f\colon M \to M'$ a monoid homomorphism, the assignment taking every $\rho \in \ff{M}X$ to $f \circ \rho \in \ff{M'}X$ defines the natural transformation
\[
	\ff{f}\colon \ff{M} \to \ff{M'}
	\text{.}
\]
Components of this transformation are injective whenever the function underlying the homomorphism $f$ is injective. It follows from basic properties of natural transformations that these weight relabelling transformations extend to the type of FuTSs. In particular, given a weight structure $\vv{M}$ and a homomorphism $f_{i,j}\colon M_{i,j} \to M'_{i,j}$ for each pair ${i,j}$ indexing $\vv{M}$ and $\vv{M'}$, we write $\vv{f}\colon \vv{M} \to \vv{M'}$ for the collection of homomorphisms $\{f_{i,j}\}$ and
$
	\ff{\vv{f}}\colon \ff{\vv{M}} \to \ff{\vv{M'}}
$
for the resulting natural transformation. Components of this transformation are injective whenever the function underlying each homomorphism $f_{i,j}$ is injective. Then, a reduction is obtained invoking \cref{thm:reduction-natural}. In order to keep the notation compact we write reduced systems obtained in this way as follows:
\begin{equation}
	\vv{f} \odot \alpha \defeq \left(\ff{\vv{f},\car(\alpha)}-\right)^{\!\vv{A}} \circ \alpha
	\text{.}
\end{equation}

\begin{lemma}
	\label{thm:futs-relabeling-weights}
	For $\vv{f}\colon \vv{M} \to \vv{M'}$, the assignment
	$
		\alpha \mapsto \vv{f} \odot \alpha
	$
	defines a full reduction going from $\cat{FuTS}\big(\vv{A},\vv{M}\big)$ to $\cat{FuTS}\big(\vv{A},\vv{M'}\big)$ whenever each $f_{i,j}$ in $\vv{f}$ is injective.
\end{lemma}

\begin{proof}
	Observe that by hypothesis on the given homomorphisms the natural transformation $\ff{\vv{f}}\colon \ff{\vv{M}} \to \ff{\vv{M'}}$ is componentwise injective. Then, the thesis follows from \cref{thm:reduction-natural}.
\end{proof}

Observe that projections of monoid products always have sections\footnote{%
	For $f$ and $g$ arrows in a category such that $g \circ f = id$, $f$ is called \emph{retraction} of $g$ and $g$ \emph{section} of $f$.
}. In particular, for $\prod M_i$, the section of the $i$-th projection is given by the assignment taking $m \in M_i$ to the tuple defined as $m$ at index $i$ and as zero elsewhere. For $\vv{M}$ consider the product of its monoids
\begin{equation}
	\prod\vv{M} = \prod_{i=0}^{n}\prod_{j=0}^{l_i} M_{i,j}
\end{equation}
and write $\vv{\iota}$ for the collection of sections $\iota_{i,j}\colon M_{i,j} \to \prod\vv{M}$ associated with it. This collection provides a weight relabelling that turn any FuTS of type $(\ff{\vv{M}}-)^{\vv{A}}$ into an one whose weights are all drawn from $\prod\vv{M}$ \ie an homogeneous FuTS. In fact, for $\alpha\colon X \to (\ff{\vv{M}}X)^{\vv{A}}$, the FuTS $\iota \odot \alpha$ is homogeneous.

\begin{lemma}
	\label{thm:futs-to-h-futs}
	The category of FuTSs fully reduces to that of homogeneous FuTSs, hence:
	\[
		\cat{FuTS} \freduceEq \cat{H-FuTS}
		\text{.}
	\]
\end{lemma}

\begin{proof}
	For each $\vv{M}$ consider $\vv{\iota}$ as above and apply \cref{thm:futs-relabeling-weights}.
\end{proof}

\noindent
We denote the full reduction described above as $\sigma_{h}\colon\cat{FuTS} \to \cat{H-FuTS}$.

\paragraph{Homogeneous nested FuTSs} 
Given any tabular homogeneous FuTSs, we observe that the components of the product
\[
	\prod_{i=0}^n\left(\ff{\vv*{M}{i}}-\right)^{A_i}
\]
defining its type are all the same (rows have the same length and monoids are the same). Therefore, the product can be regarded as an exponential and hence this type is (isomorphic to)
\[
	\left(\ff{\vv*{M}{0}}-\right)^{\prod_{i=0}^nA_i}
\]
\ie the type of an homogeneous nested FuTSs. 

\begin{lemma}
	\label{thm:ht-futs-to-hn-futs}
	The category of tabular homogeneous FuTSs fully reduces to its subcategory of nested ones, hence:
	\[
		\cat{HT-FuTS}
		\freduceEq
		\cat{HN-FuTS}
		\text{.}
	\]
\end{lemma}

\begin{proof}
	There is a natural isomorphism 
	$\prod_{i=0}^n\left(\ff{\vv*{M}{i}}-\right)^{A_i} \cong \left(\ff{\vv*{M}{0}}-\right)^{\prod_{i=0}^nA_i}$.
	The thesis follows from \cref{thm:reduction-natural}.
\end{proof}

\noindent
We denote the reduction described above as $\sigma_{n}\colon\cat{HT-FuTS} \to \cat{HN-FuTS}$.

\paragraph{Simple FuTSs}
Given an unlabelled homogeneous nested FuTS, we observe that it can be reduced to an unlabelled simple FuTS since its type is of the form $\ff{M}^{l+1}$ and meets the assumptions of \cref{thm:bisim-flattening}.

Finally, by composition of the above reductions we obtain the desired reduction taking FuTSs to unlabelled simple ones. 

\begin{theorem}
	\label{thm:futs-to-wts}
	The category of FuTSs reduces to that of unlabelled simple FuTS, hence:
	\[
		\cat{FuTS} \reduceEq \cat{WTS}\text{.}
	\]
\end{theorem}

\begin{proof}
	\setupstepsequence{step:reduction-full-abstraction-1}
	By basic properties of reductions it holds that:
	\[
		\cat{FuTS}
		\steplabel{\freduceEq}
		\cat{H-FuTS}
		\steplabel{\freduceEq}
		\cat{HT-FuTS}
		\steplabel{\freduceEq}
		\cat{HN-FuTS}
		\steplabel{\freduceEq}
		\cat{UHN-FuTS}
		\steplabel{\reduceEq}
		\cat{US-FuTS}				
	\]
	where
	\stepref{1} follows from \cref{thm:futs-to-t-futs},
	\stepref{2} follows from \cref{thm:futs-to-h-futs},
	\stepref{3} follows from \cref{thm:ht-futs-to-hn-futs},
	\stepref{4} follows from \cref{thm:futs-to-u-futs},
	and 
	\stepref{5} follows from \cref{thm:reduction-flattening}, respectively.
\end{proof}

\noindent
We denote the reduction described above as $\sigma_{s}\colon\cat{FuTS} \to \cat{WTS}$.

We remark that all but the step taking unlabelled homogeneous nested FuTS to unlabelled simple FuTSs are full reduction and hence leave state spaces of systems untouched. 

\newif\ifreduced
\newif\ifhomogeneous
\newcommand*{\drawExSys}{
	\begin{tikzpicture}[
			scale=3,
			auto,
			baseline=(current bounding box.center),
			state/.style={
				draw=black,fill=white,rounded corners=3pt,
				inner sep=3pt, outer sep=0pt},
			hidden transition/.style={->,>=latex},
			lbl/.style={outer sep=2pt,inner sep=0pt},
		]
		\ifreduced
			\tikzset{
				hidden state/.style={state},
				transition/.style={hidden transition},
			}
			\def\hs##1{##1}
		\else
			\tikzset{
				hidden state/.style={
					circle,draw=none,fill=none,
					inner sep=0pt, outer sep=-1pt},
				transition/.style={-}
			}
			\def\hs##1{}
		\fi
		\ifhomogeneous
			\def\tlbl##1{{##1, \ltrue, 0}}
			\def\htlbl##1##2{{##1, \lfalse, ##2}}
		\else
			\def\tlbl##1{##1}
			\def\htlbl##1##2{##2}
		\fi
		
		\foreach \i in {0,1,2,3}{
			\node[hidden state]
				(h\i) at (\i*90+180:1) {\hs{\(r_\i\)}};
			\node[state] (s\i) at (\i*90+135:1) {\(s_\i\)};
		}
	
		\node[hidden state] (h4) at (0,0) {\hs{\(r_4\)}};
	
		\foreach \i in {0,1,2,3}{
			\pgfmathsetmacro{\ax}{\i*90}
			\draw[transition] (s\i) .. 
				controls ([rotate=\ax]-.9, .1) 
				and ([rotate=\ax]-.35,0) .. 
				node[lbl,pos=.91] {\(\tlbl{a}\)} (h\i);
			\draw[hidden transition] (h\i) .. 
				controls ([rotate=\ax]-1.5,.25) 
				and ([rotate=\ax]-1.5,.35) .. 
				node[lbl] {\(\htlbl{a}{\frac{1}{2}}\)}(s\i);
			\pgfmathsetmacro{\j}{int(mod(\i+1,4))}
			\draw[hidden transition] (h\i) .. 
				controls ([rotate=\ax]-1.5,-.25) 
				and ([rotate=\ax]-1.5,-.35) .. 
				node[lbl,swap] {\(\htlbl{a}{\frac{1}{2}}\)}(s\j);
		}
		
		\draw[transition] (s1) .. 
			controls (-.6,-.1) and (-.1,-.5) .. 
			node[lbl,swap] {\(\tlbl{b}\)}(h4);
		\draw[hidden transition] (h4) .. 
			controls (0,.4) and (-.2,.6) .. 
			node[lbl,pos=.6] {\(\htlbl{b}{\frac{1}{6}}\)}(s0);
		\draw[hidden transition] (h4) .. 
			controls (.1,.3) and (.4,.3) .. 
			node[lbl,swap,pos=.8] {\(\htlbl{b}{\frac{1}{2}}\)}(s2);
		\draw[hidden transition] (h4) .. 
			controls (0,.4) and (.2,.6) .. 
			node[lbl,swap,pos=.6] {\(\htlbl{b}{\frac{1}{3}}\)}(s3);
	\end{tikzpicture}		
}
\newcommand*{\drawExULTraS}{\reducedfalse\drawExSys}
\newcommand*{\drawExHFuTS}{\reducedfalse\homogeneoustrue\drawExSys}
\newcommand*{\drawExWLTS}{\reducedtrue\homogeneoustrue\drawExSys}

\begin{figure}
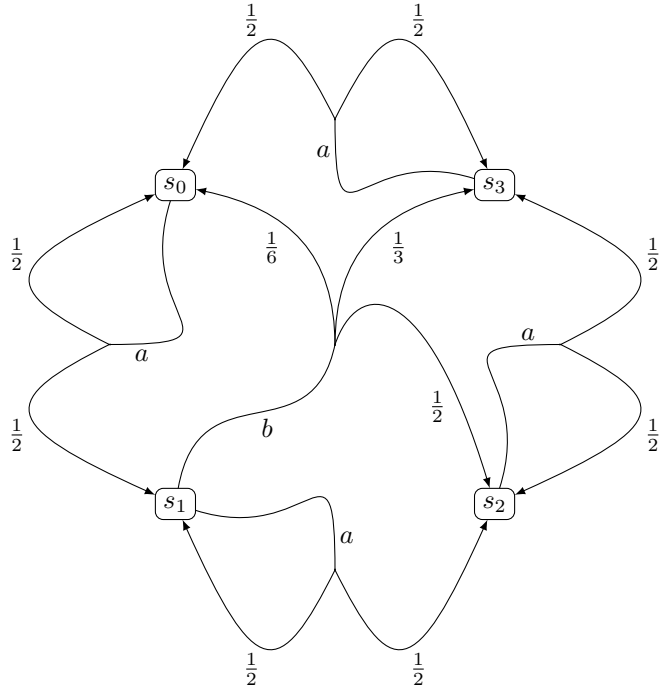

	\begin{center}
		\drawExULTraS
		\caption{An example of probabilistic ULTraS {\cite[Figure~1]{latella:lmcs2015}}.}
		\label{fig:example-ultras-reduction-1}
	\end{center}
\end{figure}

\begin{figure}
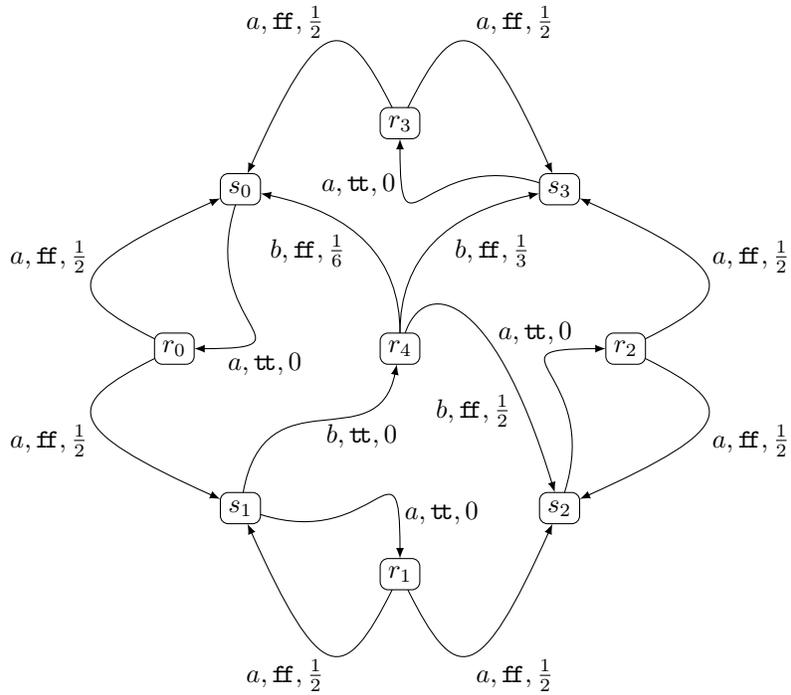

	\begin{center}
		\drawExWLTS
		\caption{The WLTS associated to the ULTraS depicted in \cref{fig:example-ultras-reduction-1}.}
		\label{fig:example-ultras-reduction-3}
	\end{center}
\end{figure}

\paragraph{{Example: reducing NPLTSs to WTSs}}
Consider the probabilistic transition system depicted in \cref{fig:example-ultras-reduction-1} and regard it as an ULTraS over the monoid of real numbers under addition. By applying the above reduction to this system we obtain the weighted transition system shown in \cref{fig:example-ultras-reduction-3}.
The function mapping states from the former system into to states of the latter is implicitly described by the names (\ie $s_0,\dots,s_3$). Likewise for the correspondence between bisimulation relations. Finally, observe that states $r_0,\dots,r_4$ correspond precisely to the reachability functions used by the first system. 

As exemplified by the above reduction, FuTSs can be reduced to WLTSs by extending the original state space with weight functions and splitting steps accordingly. From this perspective, weight functions are \emph{hidden states} in the original systems which the proposed reduction renders explicit. This observation highlights a trade-off between state and behaviour complexity of these semantically equivalent meta-models. 

\section{Fully abstract modal logics}
\label{sec:fully-abstract-logic}

In this section we consider modal logics and investigate the use of reductions as tool for proving full abstraction.
Since Hennessy-Milner's logic for LTSs \cite{hm:acm1985}, modal logics have been studied as a way to characterise the behaviour and operational semantics of transition systems. In particular, given a logic $\mathfrak{L}_{\cat{C}}$ interpreted against systems in some $\cat{C}$ (a logic for $\cat{C}$ for short), we are interested in discriminate states depending on which formulae of $\mathfrak{L}_{\cat{C}}$ they satisfy.

\begin{notation}
	For $\phi \in \mathfrak{L}_{\cat{C}}$, $\alpha \in \cat{C}$, and $x \in \car(\alpha)$, we write $x \vDash_{\alpha} \phi$ to denote that $\phi$ holds in the state $x$ of $\alpha$ and write $\llbracket \phi \rrbracket_{\alpha}$ for the set $\{x \in \vDash_{\alpha} \mid x \vDash_{\alpha} \phi \}$ of states where $\phi$ holds.
\end{notation}

\begin{definition}
	Let $\mathfrak{L}_{\cat{C}}$ a logic for systems in $\cat{C}$ a category of systems.
	For $\alpha$ in $\cat{C}$, states $x$ and $x'$ are called \emph{logical equivalent} (written $x \simeq_{\mathfrak{L}_{\cat{C}}} x'$) if and only if, for every formula $\phi$ of $\mathfrak{L}_{\cat{C}}$:
	\begin{equation}
		x \vDash_\alpha \phi \iff x' \vDash_\alpha \phi
		\text{.}
	\end{equation}
\end{definition}

\begin{notation}
In the sequel, we often write $\simeq$ instead of $\simeq_{\mathfrak{L}_{\cat{C}}}$ provided the logic is clear from the context.
\end{notation}

An important property is whether logical and bisimulation equivalence coincide for the logic and class of systems under scrutiny. Formally:

\begin{definition}
	A logic $\mathfrak{L}_{\cat{C}}$ for systems in $\cat{C}$ is called \emph{fully abstract \wrt bisimulation} provided that for every $\alpha$ in $\cat{C}$ and states $x$ and $x'$ it holds:
	\[
		x \sim x' \iff x \simeq x'
		\text{.}
	\]
\end{definition}

There are several logics for systems considered in this paper that are fully abstract. Among them we mention Hennessy-Milner's logic for LTSs \cite{hm:acm1985}, finite-disjunction logic for PLTSs \cite{bm:ictcs2016}, finite-conjunction logic for multi transition systems and Markov chains \cite{schroder:tcs2008}, finite-conjunction logic for positive WLTSs \cite{js:jlc2010}, and many more. Below we recall definitions and results about finite-conjunction logic for WLTSs that are relevant for the logic for FuTSs we introduce in \cref{sec:futs-logic-fully-abstract}. 

\paragraph{Finite-conjunction logic for WLTSs}
Finite-conjunction logic (FCL) for weighted transition systems is a minimal logic that characterises bisimulation for WLTSs whose weights are drawn from a certain class of monoids called \emph{positive} \cite{js:jlc2010,schroder:tcs2008,cp:tcs2007}.
Before we proceed to introduce syntax and interpretation of formulae for this logic let us recall that an abelian monoid $(M,+,0)$ is called \emph{positive} whenever it has the \emph{zerosumfree property}, \ie if and only if the following implication holds true:
\begin{equation}
	m + m' = 0 \implies m = 0 \land m' = 0
	\text{.}
\end{equation}
As the name suggests, positive monoids can be endowed with a partial order $\leq$ compatible with their structure in the sense that $+$ is monotonic in both components:
\begin{equation}
	m \leq m' \implies m + m'' \leq m' + m''
	\text{.}
\end{equation}
The monoidal sum induces an ordering $\trianglelefteq$ called \emph{natural} and defined as follows:
\begin{equation}
	m \trianglelefteq m' \defiff \exists m''\, m + m'' = m'
	\text{.}
\end{equation}
The natural order is the weakest\footnote{%
	A partial order is said to be weaker than another (say $\sqsubseteq$ and $\leq$, respectively) whenever it is contained by the former (${\sqsubseteq} \subseteq {\leq}$).
} of all the partial order compatible with a given positive monoid.
As a consequence, the unit $0$ is the bottom element of any ordering compatible with the structure of a monoid (hence the term ``positive'').
Examples of positive abelian monoids are $(\mathbb{B},\lor,\tt)$, $(\mathbb{N},+,0)$, $(\mathbb{N},\max,0)$, $(\mathcal{P}(X),\cup,\emptyset)$, and $([0,\infty),+,0)$. Positive monoids are closed under monoid products.
We assume every abelian monoid $(M,+,0)$ in the sequel to be positive and implicitly equipped with an ordering denoted by $\leq$, unless otherwise stated.

Finite-conjunction logic for WLTSs over positive monoids is similar to (a fragment of) Hennessy-Milner logic for CCS \cite{hm:acm1985} except for the diamond modality which is decorated with a weight lower bound. 
For $(M,+,0,\leq)$ a positive abelian monoid and $A$ a set of labels, formulae of this logic are described by the following grammar:
\begin{equation}
	\phi \Coloneqq \top \mid \phi \land \phi \mid \bflangle a|m \bfrangle \phi
\end{equation}
where $a$ and $m$ range over $A$ and $M$, respectively. 
As common practice, we will omit trailing occurrences of $\top$ and \eg write $\bflangle a|m \bfrangle$ instead of $\bflangle a|m \bfrangle\top$. When there is exactly one label in $A$ we will omit it from modalities and write just $\bflangle m \bfrangle\phi$.
Their semantics with respect to a state $x$ of a WLTS $\alpha\colon X \to (\ff{M}X)^A$ is defined as follows:
\begin{align}
	x \vDash_{\alpha} \top & \defiff \mathrm{true} \\
	x \vDash_{\alpha} \phi \land \phi' & \defiff x \vDash_{\alpha} \phi \text{ and } x \vDash_{\alpha} \phi'\\
	x \vDash_{\alpha} \bflangle a|m \bfrangle \phi & \defiff \sum_{y \in \llbracket \phi \rrbracket_{\alpha}} \alpha(x)(a)(y) \geq m
	\text{.}
\end{align}
It follows from the positiveness assumption on $M$ that the formula $\bflangle a| 0 \bfrangle \phi$ is satisfied by any state of any WLTS, regardless of $\phi$.

In \cite{js:jlc2010} it is shown FCL for WLTSs is fully abstract \wrt bisimulation provided that weights are drawn from positive monoids with the \emph{cancellation} property \ie whenever their structure satisfies the implication:
\begin{equation}
	m + m' = m + m'' \implies m' = m''
	\text{.}
\end{equation}
Examples of cancellative abelian monoids are $(\Sigma^*,\cdot,\varepsilon)$, $(\mathbb{N},+,0)$, $([0,\infty),+,0)$, and $(\mathbb{R},+,0)$. Cancellative monoids are closed under monoid products.

\begin{theorem}[{\cite[Thm~13]{js:jlc2010}}]
	\label{thm:wlts-logic-bisimulation}
	For $(X,\alpha)$ a WLTS with weights drawn from a positive cancellative monoid, 
	\begin{equation}
		x \sim x' \iff x \simeq x'
		\text{.}
	\end{equation}
\end{theorem}

\section{Reductions and fully abstract logics}
\label{sec:reductions-logic}

In order to relate categories of systems equipped with a notion of logical equivalence mimicking reductions we introduce the notion of \emph{translation}. Intuitively, translations are reductions except that logical equivalence is considered instead of bisimulations. 

\begin{definition}
	\label{def:translation}
	Let $\mathfrak{L}_{\cat{C}}$ and $\mathfrak{L}_{\cat{D}}$ logics be for systems in $\cat{C}$ and $\cat{D}$, respectively.
	A translation from $\mathfrak{L}_{\cat{C}}$ to $\mathfrak{L}_{\cat{D}}$ is given by
	\begin{itemize}
		\item a function $\theta\colon \mathfrak{L}_{\cat{C}} \to \mathfrak{L}_{\cat{D}}$,
		\item a functor $\theta\colon \cat{C} \to \cat{D}$,
		\item an injective function $\theta^{c}_{\alpha}\colon \car(\alpha) \to \car(\theta(\alpha))$ for any $\alpha \in \cat{C}$
	\end{itemize}
	with the following properties:
	\begin{enumerate}
		\item \label{def:translation-c-homomorphisms}
			for any $f\colon \alpha \to \beta$ in $\cat{C}$, $\theta_{\beta}^{c} \circ f = \theta(f) \circ \theta_{\alpha}^{c}$;
		\item  \label{def:translation-c-semantics}
			for any system $\alpha$ in $\cat{C}$, state $x$ of $\alpha$ and formula $\phi$ in $\mathfrak{L}_{\cat{C}}$,
			$
				x \vDash_{\alpha} \phi \iff \theta^{c}_{\alpha}(x) \vDash_{\theta(\alpha)}  \theta(\phi)
				\text{;}$
		\item  \label{def:translation-c-logical-equivalence}
			for any system $\alpha$ in $\cat{C}$ and states $x$, $x'$ of $\alpha$,
			$x \simeq_{\mathfrak{L}_{\cat{C}}} x' \iff \theta^{c}_{\alpha}(x) \simeq_{\mathfrak{L}_{\cat{D}}} \theta^{c}_{\alpha}(x')
				\text{.}$
	\end{enumerate}
\end{definition}

Condition~\ref{def:translation-c-homomorphisms} of \cref{def:translation} states that translations are coherent with the structure of homomorphisms similarly to \eqref{eq:reduction-c-homomorphism} which states the same coherency condition for reductions. Condition~\ref{def:translation-c-semantics} states that the translation is coherent with formulae semantics. Condition~\ref{def:translation-c-logical-equivalence} states that translations preserve and reflect logical equivalence. Note that translations need not to be surjective on formulae and hence Condition~\ref{def:translation-c-semantics} does not entail Condition~\ref{def:translation-c-logical-equivalence}.

A translation is said to be \emph{coherent} with a reduction whenever they share their underlying functor and all injective maps are also system reductions. Whenever a translation is coherent with a reduction we call the pair a reduction for the logics involved. Formally:

\begin{definition}
	 A reduction from $\mathfrak{L}_{\cat{C}}$ to $\mathfrak{L}_{\cat{D}}$ is given by
	 \begin{itemize}
 		\item a function $\theta\colon \mathfrak{L}_{\cat{C}} \to \mathfrak{L}_{\cat{D}}$,
 		\item a functor $\theta\colon \cat{C} \to \cat{D}$,
 		\item a system reduction $\theta^{c}_{\alpha}\colon \car(\alpha) \to \car(\theta(\alpha))$ for any $\alpha \in \cat{C}$
 	\end{itemize}
 	subject to Conditions~\ref{def:translation-c-homomorphisms} to~\ref{def:translation-c-logical-equivalence} from \cref{def:translation}.
\end{definition}

In the sequel we often consider reductions for logics obtained equipping a reduction with suitable functions taking formulae to their translations. Therefore, we adopt the convention of writing a reduction from $\mathfrak{L}_{\cat{C}}$ to $\mathfrak{L}_{\cat{D}}$ as a pair $(\sigma,\theta)$ where $\sigma\colon \cat{C} \to \cat{D}$ is a reduction and is $\theta\colon \mathfrak{L}_{\cat{C}} \to \mathfrak{L}_{\cat{D}}$ is a function that extends to a translation coherent with $\sigma$.
We extend notation and terminology introduced for reductions to this settings. A reduction $(\sigma,\theta)$ is called \emph{full} whenever $\sigma$ is so. A logic $\mathfrak{L}_{\cat{C}}$ is said to \emph{reduce} (resp.~\emph{fully} reduce) to $\mathfrak{L}_{\cat{D}}$ provided there is a reduction (resp.~\emph{full} reduction) going from the former to the latter. 

The notion of reduction for logics allows us to combine the information translations and reduction carry with regard to logical and bisimulation equivalences. As a consequence, reductions for logics reflect the property of full abstraction.

\begin{theorem}
	\label{thm:reduction-full-abstraction}
	Assume that $\mathfrak{L}_{\cat{C}}$ reduces to $\mathfrak{L}_{\cat{D}}$. 
	If $\mathfrak{L}_{\cat{D}}$ is fully abstract so is $\mathfrak{L}_{\cat{C}}$.
\end{theorem}

\begin{proof}
	\setupstepsequence{step:reduction-full-abstraction}
	Let $\theta$ be a reduction going from $\mathfrak{L}_{\cat{C}}$ to $\mathfrak{L}_{\cat{D}}$.
	For every $\alpha$ in $\cat{C}$ and states $x$ and $x'$ it holds that
	\[
		x \sim x'
		\steplabel{\iff}
		\sigma^{c}_{\alpha}(x) \sim \sigma^{c}_{\alpha}(x')
		\steplabel{\iff}
		\sigma^{c}_{\alpha}(x) \simeq_{\mathfrak{L}_{\cat{D}}} \sigma^{c}_{\alpha}(x')
		\steplabel{\iff}
		x \simeq_{\mathfrak{L}_{\cat{C}}} x'
	\]
	where \stepref{1}, \stepref{2}, and \stepref{3} follow from definition of system reduction, of full abstraction \wrt bisimulation, and of translation, respectively.
\end{proof}

In other words, \cref{thm:reduction-full-abstraction} introduces a technique for proving that a logic is fully abstract: reduce it to one that is known to have this property.

\section{Application: a fully abstract modal logic for FuTSs}
\label{sec:futs-logic-fully-abstract}

We introduce \emph{finite-conjunction logic} for FuTSs whose weights are drawn from positive monoids. This is a conservative extension of finite-conjunction logic for WLTSs where the diamond modality is decorated with a sequence of weight lower bounds--one for each sub-step forming a FuTS transition.
For labels $\vv{A} = \langle A_0,\dots,A_n\rangle$ and monoids $\vv{M} = \langle\langle M_{0,0},\dots,M_{0,l_0}\rangle,\dots,\langle M_{n,0} \dots M_{n,l_n}\rangle\rangle$, formulae of this logic are described by the following grammar:
\begin{equation}
	\phi \Coloneqq 
		\top \mid 
		\phi \land \phi \mid 
		\bflangle i|a|\vv{m} \bfrangle\phi
\end{equation}
where $i \in \{0,\dots,n\}$ indexes the $n+1$ components of FuTSs considered and $a$ and $\vv{m}$ range over the set $A_i$ of labels and sequence $\vv{M_i}$ of weighing monoids for the $i$-th component of FuTSs considered, respectively.
We adopt the same syntactic conventions for the FCL for WLTSs and omit the component index $i$ from modalities in presence of nested FuTSs (\ie whenever $n = 0$). The latter convention allows us to regard formulae of the logic for WLTSs as formulae of FCL for FuTSs just defined. 

\begin{notation}
	For $\cat{C}$ a subcategory of $\cat{FuTS}$ we write $\mathfrak{L}^{\land}_{\cat{C}}$ for the set of FCL formulae that can be interpreted against systems in $\cat{C}$. For instance, $\mathfrak{L}^{\land}_{\cat{WLTS}}$ is the set of formulae of FCL for WLTSs.
\end{notation}

Before we define how formulae of this logic are interpreted, let us introduce some auxiliary notion. For $m$ an element of a positive monoid $M$ and $Y \subseteq X$, write $\bflangle m \bfrangle Y$ for the set of weight functions assigning to $Y$ a weight above $m$:
\[
	\bflangle m \bfrangle Y \defeq \left\{\rho \in \ff{M}X \,\middle|\, \sum_{y \in Y} \rho(y) \geq m\right\}
	\text{.}
\]
The assignment $Y \mapsto \bflangle m \bfrangle Y$ defines the monotonic operator $\bflangle m \bfrangle\colon \mathcal{P}X \to \mathcal{P}\ff{M}X$.

Formulae semantics with respect to a state $x$ of a FuTS $\alpha$ is defined as follows:
\begin{align}
	x \vDash \top & \defiff \mathrm{true} \\
	x \vDash \phi \land \phi' & \defiff x \vDash \phi \text{ and } x \vDash \phi'\\
	x \vDash \bflangle i|a|m_{0}, \dots, m_{l_i} \bfrangle \phi & \defiff 
		\alpha_{i}(x)(a) \in  
	\bflangle m_0 \bfrangle\dots\bflangle m_{l_i} \bfrangle \llbracket \phi \rrbracket_{\alpha}\text{.}
\end{align}
Note that in the case of simple FuTSs (\ie WLTSs), this semantics is precisely that described in the previous section: finite-conjunction logic for FuTSs is a conservative extension of finite-conjunction logic for WLTSs.

We claim that finite-conjunction logic for FuTSs is fully abstract with respect to bisimulation provided that weights are drawn from positive cancellative monoids. In order to prove this result, we extend the sequence of reductions we presented in \cref{sec:futs-reductions} with coherent translation for formulae involved and then invoke \cref{thm:reduction-full-abstraction}.

Before we proceed let us discuss some useful properties of the operator $\bflangle m \bfrangle\colon \mathcal{P}X \to \mathcal{P}\ff{M}X$ introduced above. As stated by \cref{thm:inverse-modality-props} below, this operator distributes over arbitrary intersections, respects projections, and commutes with homomorphisms.
\begin{lemma}
	\label{thm:inverse-modality-props}
	For $X$ a set, the following statements hold:
	\begin{enumerate}
	\item \label{thm:inverse-modality-props-1}
		for $M$ a positive monoid, and $\{Y_i \subseteq X\}_{i \in I}$ a family of subsets of $X$:
		\[
			\bflangle m \bfrangle \bigcap_{i \in I} Y_{i} 
			= 
			\bigcap_{i \in I} \bflangle m \bfrangle Y_{i}
			\text{;}
		\]
	\item \label{thm:inverse-modality-props-2}
		for $M$ a positive monoid, $\{m_0,\dots,m_n\}$ a finite family of weights, and $Y \subseteq X$:
		\[
			\bflangle m_0 + \dots + m_n \bfrangle Y 
			= 
			\bigcap_{i = 0}^{n} \bflangle m_i \bfrangle Y
			\text{;}
		\]
	\item \label{thm:inverse-modality-props-3}
		for $\{M_i\}_{i \in I}$ non-empty family of positive monoids  and $Y \subseteq X$:
		\[
			\bflangle \{m_i\}_{i \in I} \bfrangle Y = \bigcap_{i \in I} \bflangle \overline{m_i} \bfrangle Y
		\]
		where $\overline{m_i} \in \prod_{i \in I} M_i$ has value $m_i$ at $i$ and zero elsewhere;
	\item \label{thm:inverse-modality-props-4}
		for $f\colon M \to M'$ a homomorphism of positive monoids  and $Y \subseteq X$:
		\[
			\bflangle f(m) \bfrangle Y 
			=
			\{f \circ \rho \mid \rho \in \bflangle m \bfrangle Y \} 
			=
			\{\ff{f,X}(\rho) \mid \rho \in \bflangle m \bfrangle Y \}
			\text{;}
		\]
	\end{enumerate}
\end{lemma}

\paragraph{Unlabelled FuTSs} 
Consider the function $\theta_{u}\colon \mathfrak{L}^{\land}_{\cat{FuTS}} \to \mathfrak{L}^{\land}_{\cat{U-FuTS}}$ defined, on each $\vv{M}$ and $\vv{A}$, as follows:
\begin{align}
	\theta_{u}(\top) & \defeq \top \\
	\theta_{u}(\phi \land \phi') & \defeq \theta_{u}(\phi) \land \theta_{u}(\phi') \\
	\theta_{u}(\bflangle i|a| m_0, \dots, m_{l_i} \bfrangle \phi) & \defeq
	\bflangle i | a \cdot m_0, \dots, m_{l_i} \bfrangle \theta_{u}(\phi)
	\text{.}
\end{align}
Below we prove that $\theta_{u}$ is coherent with the full reduction $\sigma_{u}\colon\cat{FuTS} \to \cat{U-FuTS}$ defined in the proof of \cref{thm:futs-to-u-futs}. As a consequence, the pair $(\sigma_{u},\theta_{u})$ defines a full reduction going from $\mathfrak{L}^{\land}_{\cat{FuTS}}$ to $\mathfrak{L}^{\land}_{\cat{U-FuTS}}$.

\begin{lemma}
	\label{thm:futs-logic-u-futs}
	FCL for FuTSs fully reduces to that for unlabelled ones, hence:
	\[
		\mathfrak{L}^{\land}_{\cat{FuTS}} \freduceEq \mathfrak{L}^{\land}_{\cat{U-FuTS}}
		\text{.}
	\]
\end{lemma}

\begin{proof}
	First we prove that $\theta_{u}$ and $\sigma_{u}$ meet Condition~\ref{def:translation-c-semantics} of \cref{def:translation} \ie that for any $\phi \in \mathfrak{L}^{\land}_{\cat{FuTS}}$, $\alpha \in \cat{FuTS}$, and $x \in \car(\alpha)$:
	\[
		x \vDash_{\alpha} \phi 
		\iff 
		x \vDash_{\sigma_{u}(\alpha)} \theta_{u}(\phi)
	\]
	(note that $\sigma^{c}_{u,\alpha}(x) = x$.)
	We proceed by recursion on the structure of $\phi$.
	\begin{description}
		\item[$\top$] 
			Let $\phi = \top$. Clearly $x \vDash_{\alpha} \phi$ and $x \vDash_{\sigma_{u}(\alpha)} \theta_{u}(\phi)$. 
		\item[$\land$] 
			Let $\phi = \phi' \land \phi''$. Assume by induction hypothesis that $\phi'$ and $\phi''$ hold at $x$ if and only if $\theta_{u}(\phi')$ and $\theta_{u}(\phi'')$ hold at $x$, respectively. As a consequence $x \vDash_{\alpha} \phi' \land \phi''$ if and only if $x \vDash_{\sigma_{u}(\alpha)} \theta_{u}(\phi' \land \phi'')$.
		\item[$\bflangle-\bfrangle$] 
			Let $\phi = \bflangle i | a | m_{0},\dots,m_{l_i}\bfrangle \phi'$. 
			It follows from fullness of $\sigma_{u}$ and induction hypothesis that
			\[
				\llbracket \phi' \rrbracket_{\alpha} = 
				\llbracket \theta_{u}(\phi') \rrbracket_{\sigma(\alpha)}
			\]
			and hence that
			\[
				\bflangle m_{1}\bfrangle\dots\bflangle m_{l_i}\bfrangle
				\llbracket \phi' \rrbracket_{\alpha} = 
				\bflangle m_{1}\bfrangle\dots\bflangle m_{l_i}\bfrangle
				\llbracket \theta_{u}(\phi') \rrbracket_{\sigma(\alpha)}
				\text{.}
			\]
			Write $Y$ for this set.
			Recall from \cref{sec:futs-reductions} that $\sigma_{u}(\alpha)_i(x)$ is the function $\sum_{\rho} (a \cdot \alpha_i(x)(a)(\rho)) \cdot \rho$ \ie that
			$
				\alpha_i(x)(a)(\rho) = m
				\iff
				\sigma_{u}(\alpha)_i(x)(\rho) = a \cdot m 
				\text{.}
			$
			We conclude that 
			$
				\alpha_i(x)(a) \in \bflangle m_0\bfrangle Y
				\iff
				\sigma_{u}(\alpha)_i(x) \in \bflangle a \cdot m_0\bfrangle Y
			$
			and hence that 
			$x \vDash_{\alpha} \bflangle i | a | m_0, \dots, m_{l_i} \bfrangle \theta_{u}(\phi') \iff x \vDash_{\sigma_{u}(\alpha)} \bflangle i | a \cdot m_0, \dots, m_{l_i} \bfrangle \theta_{u}(\phi')$.
	\end{description}
	We observe that $\theta_{u}$ is not surjective as a consequence of how labels are encoded using weights: for instance, the formula
	$\bflangle i | \lambda a. m_{0},\dots,m_{l_i}\bfrangle$ is not in the image of $\theta_{u}$ whenever $A$ has infinitely many labels. If FCL had arbitrary conjunctions then, the previous formula would have been equivalent to an $A$ indexed conjunction of formulae in the image of $\theta_{u}$, namely $\bigvee_{a \in A} \bflangle i | a \cdot m_{0},\dots,m_{l_i}\bfrangle$. From these observations we conclude that to prove that $\theta_{u}$ and $\sigma_{u}$ meet Condition~\ref{def:translation-c-logical-equivalence} of \cref{def:translation} it suffices to show that for every formula $\psi \in \mathfrak{L}^{\land}_{\cat{U-FuTS}}$ there is a set of formulae $\{\phi_i\}_{j \in J}$ with the property that for every $\alpha \in \cat{FuTS}$ and $x \in \car(\alpha)$:
	\[
		x \vDash_{\sigma_{u}(\alpha)} \psi \iff \forall j \in J \left(x \vDash_{\alpha} \phi_j\right)
	\]
	which can be equivalently written as
	\[
		\llbracket \psi \rrbracket_{\sigma(\alpha)} = \bigcap_{j \in J} \llbracket \phi_j \rrbracket_{\alpha}
		\text{.}
	\]
	We proceed by recursion on the structure of $\psi$.
	\begin{description}
		\item[$\top$] 
			Let $\psi = \top$. Then, $\psi \in \img(\theta_{u})$.
		\item[$\land$] 
			Let $\psi = \psi' \land \psi''$. Assume induction hypothesis that $\{\phi'_{j}\}_{j \in J}$ and  $\{\phi''_{k}\}_{k \in K}$ for $\phi'$ and $\psi''$ are given. Consider the set of formulae $\{\phi'_{j} \land \phi''_{k}\}_{(j,k) \in J \times K}$. By construction it holds that: $x \vDash_{\sigma_{u}(\alpha)} \psi'\land \psi'' \iff \forall (j,k) \in J \times K \left(x \vDash_{\alpha} \phi'_{j} \land \phi''_{k} \right)$.
		\item[$\bflangle-\bfrangle$] 
			Let $\psi = \bflangle i | m_{0},\dots,m_{l_i}\bfrangle \phi'$. 
			Assume by induction hypothesis that the set $\{\phi'_{j}\}_{j \in J}$ for $\psi'$ has been given. 
			Consider the set of formulae $\{\phi_{j,a}\}_{(j,a) \in J \times A}$ where each $\phi_{j,a}$ is $\bflangle i | a | m_{0}(a),\dots,m_{l_i}\bfrangle \phi'_j$.
			It follows from fullness of $\sigma_u$, induction hypothesis, and \cref{thm:inverse-modality-props} (\cref{thm:inverse-modality-props-1}) that:
			\begin{equation}
				\bflangle m_{1}\bfrangle\dots\bflangle m_{l_i}\bfrangle
				\llbracket \psi' \rrbracket_{\sigma_{u}(\alpha)}
				= 
				\bflangle m_{1}\bfrangle\dots\bflangle m_{l_i}\bfrangle
				\bigcap_{j \in J}\llbracket \phi'_j \rrbracket_{\alpha}
				= 
				\bigcap_{j \in J}
				\bflangle m_{1}\bfrangle\dots\bflangle m_{l_i}\bfrangle
				\llbracket \phi'_j \rrbracket_{\alpha}.
			\end{equation}
			It follows from \cref{thm:inverse-modality-props} (\cref{thm:inverse-modality-props-1,thm:inverse-modality-props-3}) that
			\[
				\bflangle m_{0}\bfrangle\dots\bflangle m_{l_i}\bfrangle
				\llbracket \psi' \rrbracket_{\sigma_{u}(\alpha)} = 
				\bigcap_{a \in A}\bigcap_{j \in J}
				\bflangle m_{0}(a)\bfrangle\bflangle m_{1}\bfrangle\dots\bflangle m_{l_i}\bfrangle
				\llbracket \phi'_j \rrbracket_{\alpha}.
			\]
			We conclude from the above and the semantics of FCL that
			$\llbracket \psi \rrbracket_{\sigma(\alpha)} = \bigcap_{(a,j) \in A\times J} \llbracket \phi_(a,j) \rrbracket_{\alpha}$.
			\qedhere			
	\end{description}
\end{proof}

\paragraph{Tabular FuTSs} 
Fix a non-zero element $\mathtt{p}$ of a positive monoid $\mathbb{P}$ (later we will also assume cancellation) and consider the function $\theta_{t}\colon \mathfrak{L}^{\land}_{\cat{FuTS}} \to \mathfrak{L}^{\land}_{\cat{T-FuTS}}$ defined, on each tabular $\vv{M}$ and $\vv{A}$, as follows:
\begin{align}
	\theta_{t}(\top) & \defeq \top \\
	\theta_{t}(\phi \land \phi') & \defeq \theta_{t}(\phi) \land \theta_{t}(\phi') \\
	\theta_{t}(\bflangle i|a|m_0, \dots, m_{l_i} \bfrangle \phi) & \defeq
	\bflangle i | a | \underbrace{\mathtt{p}, \dots, \mathtt{p}}_{l-l_i\text{ times}}, m_0, \dots, m_{l_i} \bfrangle \theta_{t}(\phi)\text{.}
\end{align}
Below we prove that $\theta_{t}$ is coherent with the full reduction $\sigma_{t}\colon\cat{FuTS} \to \cat{T-FuTS}$ defined in the proof of \cref{thm:futs-to-t-futs}. As a consequence, the pair $(\sigma_{t},\theta_{t})$ defines a full reduction going from $\mathfrak{L}^{\land}_{\cat{FuTS}}$ to $\mathfrak{L}^{\land}_{\cat{T-FuTS}}$.

\begin{lemma}
	\label{thm:futs-logic-t-futs}
	FCL for FuTSs fully reduces to that for tabular ones, hence:
	\[
		\mathfrak{L}^{\land}_{\cat{FuTS}} \freduceEq \mathfrak{L}^{\land}_{\cat{T-FuTS}}
		\text{.}
	\]
\end{lemma}

\begin{proof}
	Recall that $\sigma_{t}(\alpha)$ introduces functional steps ($\mathtt{p}$-valued Dirac's delta functions like $\mathtt{p} \cdot x$) in correspondence all instances of $\mathbb{P}$ introduced by $[\vv{M}]$. We observe that 
	\[
		y \in Y 
		\iff
		\mathtt{p} \cdot y \in \bflangle \mathtt{p} \bfrangle Y
		\text{.}
	\]
	The proof proceeds along the lines of that of \cref{thm:futs-logic-u-futs} and straightforward induction on the structure of formulae.	
\end{proof}

\paragraph{Homogeneous FuTSs} 
Consider the function $\theta_{h}\colon \mathfrak{L}^{\land}_{\cat{FuTS}} \to \mathfrak{L}^{\land}_{\cat{H-FuTS}}$ defined, on each homogeneous $\vv{M}$ and $\vv{A}$, as follows:
\begin{align}
	\theta_{h}(\top) & \defeq \top \\
	\theta_{h}(\phi \land \phi') & \defeq \theta_{h}(\phi) \land \theta_{h}(\phi') \\
	\theta_{h}(\bflangle i|a|m_0, \dots, m_{l_i} \bfrangle \phi) & \defeq
	\bflangle i | a | \iota_{i,0}(m_0), \dots, \iota_{i,l_i}(m_{l_i}) \bfrangle \theta_{h}(\phi)\text{.}
\end{align}
Below we prove that $\theta_{h}$ is coherent with the full reduction $\sigma_{h}\colon\cat{FuTS} \to \cat{H-FuTS}$ defined in the proof of \cref{thm:futs-to-u-futs}. As a consequence, the pair $(\sigma_{h},\theta_{h})$ defines a full reduction going from $\mathfrak{L}^{\land}_{\cat{FuTS}}$ to $\mathfrak{L}^{\land}_{\cat{H-FuTS}}$.

\begin{lemma}
	\label{thm:futs-logic-h-futs}
	FCL for FuTSs fully reduces to that for homogeneous ones, hence:%
	\[
		\mathfrak{L}^{\land}_{\cat{FuTS}} \freduceEq \mathfrak{L}^{\land}_{\cat{H-FuTS}}
		\text{.}
	\]
\end{lemma}

\begin{proof}
	Recall from \cref{sec:futs-reductions} that $\sigma_{h}(\alpha)$ relabels weights by means of natural transformations $\ff{\iota_{i,j}}\colon \ff{M_{i,j}} \to \ff{\prod \vv{M}}$ induced by sections of the projections associated to the product of monoids $\prod \vv{M}$. 
	It follows from \cref{thm:inverse-modality-props} (\cref{thm:inverse-modality-props-4}) that for any value $m \in M_{i,j}$ and set $X$, the component $\ff{\iota_{i,j},X}\ff{M_{i,j}}X \to \ff{\prod \vv{M}}X$ takes each weight function $\rho \in \bflangle m \bfrangle Y$ to $\iota_{i,j} \circ \rho \in \bflangle \iota_{i,j}(m) \bfrangle Y$. 
	As a consequence of these observations, the proof proceeds along the lines of that of \cref{thm:futs-logic-u-futs} and straightforward induction on the structure of formulae.	
\end{proof}

\paragraph{Homogeneous nested FuTSs} 
Consider the function $\theta_{n}\colon \mathfrak{L}^{\land}_{\cat{HT-FuTS}} \to \mathfrak{L}^{\land}_{\cat{HN-FuTS}}$ defined, on each homogeneous and tabular $\vv{M}$ and $\vv{A}$, as follows:
\begin{align}
	\theta_{n}(\top) & \defeq \top \\
	\theta_{n}(\phi \land \phi') & \defeq \theta_{n}(\phi) \land \theta_{n}(\phi') \\
	\theta_{n}(\bflangle i|a| \vv{m} \bfrangle \phi) & \defeq
	\bflangle (i,a) | \vv{m} \bfrangle \theta_{n}(\phi)
	\text{.}
\end{align}
Below we prove that $\theta_{n}$ is coherent with the full reduction $\sigma_{n}\colon\cat{HT-FuTS} \to \cat{HN-FuTS}$ defined in the proof of \cref{thm:ht-futs-to-hn-futs}. As a consequence, the pair $(\sigma_{n},\theta_{n})$ defines a full reduction going from $\mathfrak{L}^{\land}_{\cat{HT-FuTS}}$ to $\mathfrak{L}^{\land}_{\cat{HN-FuTS}}$.
\begin{lemma}
	\label{thm:ht-futs-logic-hn-futs}
	FCL for tabular homogeneous FuTSs fully reduces to that for homogeneous nested ones, hence:%
	\[
		\mathfrak{L}^{\land}_{\cat{HT-FuTS}} \freduceEq \mathfrak{L}^{\land}_{\cat{HN-FuTS}}
		\text{.}
	\]
\end{lemma}

\begin{proof}
	Recall from \cref{sec:futs-reductions} that $\sigma_{n}(\alpha)(x)(i,a) = \alpha_i(x)(a)$. As a consequence, it holds that
	\[
		\alpha_i(x)(a) \in \bflangle m \bfrangle Y 
		\iff 
		\sigma_{n}(\alpha)(x)(i,a) \in \bflangle m \bfrangle Y 
		\text{.}
	\]
	Then, the proof follows the same argument used in the proof of \cref{thm:futs-logic-u-futs}.
\end{proof}

\paragraph{Simple FuTSs} 
Consider the function $\theta_{s}\colon \mathfrak{L}^{\land}_{\cat{UHN-FuTS}} \to \mathfrak{L}^{\land}_{\cat{US-FuTS}}$ defined, on each homogeneous $\vv{M}$, as follows:
\begin{align}
	\theta_{s}(\top) & \defeq \top \\
	\theta_{s}(\phi \land \phi') & \defeq \theta_{s}(\phi) \land \theta_{s}(\phi') \\
	\theta_{s}(\bflangle m_0,\dots,m_l \bfrangle \phi) & \defeq
	\bflangle m_0 \bfrangle \dots \bflangle m_l \bfrangle \theta_{s}(\phi)
	\text{.}
\end{align}
Below we prove that $\theta_{s}$ is coherent with the reduction $\sigma_{s}\colon\cat{UHN-FuTS} \to \cat{US-FuTS}$ obtained from \cref{thm:reduction-flattening}. As a consequence, the pair $(\sigma_{s},\theta_{s})$ defines a full reduction going from $\mathfrak{L}^{\land}_{\cat{UHN-FuTS}}$ to $\mathfrak{L}^{\land}_{\cat{US-FuTS}}$.
\begin{lemma}
	\label{thm:uhn-futs-logic-wts}
	FCL for unlabelled homogeneous nested FuTSs reduces to that for unlabelled simple ones, hence:%
	\[
		\mathfrak{L}^{\land}_{\cat{UHN-FuTS}} \reduceEq \mathfrak{L}^{\land}_{\cat{US-FuTS}}
		\text{.}
	\]
\end{lemma}

\begin{proof}
	In order to prove that $\theta_s$ and $\sigma_{s}$ satisfy Condition~\ref{def:translation-c-semantics} of \cref{def:translation} we proceed by structural induction and show that for $\phi \in \mathfrak{L}^{\land}_{\cat{UHN-FuTS}}$ and $\alpha \in \cat{UHN-FuTS}$ 
	$
		\llbracket \phi \rrbracket_{\alpha} 
		= 
		\car(\alpha) \cap \llbracket \theta(\phi) \rrbracket_{\sigma_s(\alpha)}
	$.
	\begin{description}
		\item[$\top$] 
			Let $\phi = \top$. Recall that $\car(\alpha) \subseteq \car(\sigma_s(\alpha))$ and $\llbracket \top \rrbracket_{\alpha} = 	\car(\alpha)$. Then, $\llbracket \top \rrbracket_{\alpha} = \car(\alpha) \cap \llbracket \top \rrbracket_{\sigma_s(\alpha)}$.
		\item[$\land$] 
			Let $\phi = \phi' \land \phi''$. By the semantics of $\land$ and induction hypothesis:
			\begin{equation}
				\llbracket \phi' \land \phi'' \rrbracket_{\alpha} 
				=
				\llbracket \phi' \rrbracket_{\alpha} \cap \llbracket \phi'' \rrbracket_{\alpha} 
				= 
				\car(\alpha) \cap \llbracket \theta(\phi') \rrbracket_{\sigma_s(\alpha)} \cap \llbracket \theta(\phi'') \rrbracket_{\sigma_s(\alpha)}
				=
				\car(\alpha) \cap \llbracket \theta(\phi') \land \theta(\phi'') \rrbracket_{\sigma_s(\alpha)}
				\text{.}
			\end{equation}
		\item[$\bflangle-\bfrangle$] Let $\phi = \bflangle m_0,\dots,m_l \bfrangle \phi'$.  By the semantics of $\bflangle-\bfrangle$ and induction hypothesis:
			\begin{equation}
				\llbracket\bflangle m_0,\dots,m_l \bfrangle \phi'\rrbracket_{\alpha}
				=
				\bflangle m_0\bfrangle ,\dots,\bflangle m_l \bfrangle \llbracket\phi'\rrbracket_{\alpha}
				=
				\bflangle m_0\bfrangle ,\dots,\bflangle m_l \bfrangle \left(\car(\alpha) \cap \llbracket \theta(\phi') \rrbracket_{\sigma_s(\alpha)}\right)
				\text{.}
			\end{equation}
		Recall from \cref{thm:reduction-flattening} that if $x \in \car(\alpha)$ then, $\supp(\sigma_s(\alpha)(x)) \subseteq \car(\alpha)$.
		Therefore, the set $\bflangle m_0\bfrangle ,\dots,\bflangle m_l \bfrangle \left(\car(\alpha) \cap \llbracket \theta(\phi') \rrbracket_{\sigma_s(\alpha)}\right)$ equals to the set
		$
			\car(\alpha) \cap \bflangle m_0\bfrangle ,\dots,\bflangle m_l \bfrangle\llbracket \theta(\phi') \rrbracket_{\sigma_s(\alpha)} 
			\text{.}
		$
		We conclude that 
		\begin{equation}
			\llbracket\bflangle m_0,\dots,m_l \bfrangle \phi'\rrbracket_{\alpha}
			=
			\car(\alpha) \cap \bflangle m_0\bfrangle ,\dots,\bflangle m_l \bfrangle\llbracket \theta(\phi') \rrbracket_{\sigma_s(\alpha)}
			= 
			\car(\alpha) \cap \llbracket\bflangle m_0\bfrangle ,\dots,\bflangle m_l \bfrangle \theta(\phi') \rrbracket_{\sigma_s(\alpha)}
			\text{.}
		\end{equation}
	\end{description}
	As a consequence of Condition~\ref{def:translation-c-semantics},
	to prove that $\theta_s$ and $\sigma_{s}$ satisfy Condition~\ref{def:translation-c-logical-equivalence} it suffice to prove that for every $\psi \in \mathfrak{L}^{\land}_{\cat{US-FuTS}}$ there is $\phi \in \mathfrak{L}^{\land}_{\cat{UHN-FuTS}}$ whose translation is equivalent to $\psi$ \ie:
	$\llbracket \psi \rrbracket_{\sigma_{s}(\alpha)}=\llbracket \theta_{s}(\phi) \rrbracket_{\sigma_{s}(\alpha)}$. This is readily achieved thanks to \cref{thm:inverse-modality-props} (\cref{thm:inverse-modality-props-2}) since it allows to distribute $\bflangle - \bfrangle$ over $\land$. In fact, for every $\alpha \in \cat{UHN-FuTS}$, we have that: 
	\begin{align}
		\llbracket\bflangle m \bfrangle \left(\phi \land \phi'\right)\rrbracket_{\alpha}
		& =
		\bflangle m \bfrangle \llbracket\phi \land \phi'\rrbracket_{\alpha}
		= 
		\bflangle m \bfrangle \left(\llbracket\phi\rrbracket_{\alpha} \cap \llbracket\phi'\rrbracket_{\alpha}\right)
		=
		\bflangle m \bfrangle\llbracket\phi\rrbracket_{\alpha} \cap 
		\bflangle m \bfrangle\llbracket\phi'\rrbracket_{\alpha}
		=
		\llbracket\bflangle m \bfrangle\phi\rrbracket_{\alpha} \cap 
		\llbracket\bflangle m \bfrangle\phi'\rrbracket_{\alpha}
		\\ & =
		\llbracket
			\bflangle m \bfrangle\phi \land  
			\bflangle m \bfrangle\phi'
		\rrbracket_{\alpha}
		\text{.}
	\end{align}
	Finally, a formula $\bigwedge_{i = 0}^{n} \bflangle m_{i,0} \bfrangle \dots \bflangle m_{i,l_i} \bfrangle $ is equivalent to $\bigwedge_{i = 0}^{n} \bflangle m_{i,0} \bfrangle \dots \bflangle m_{i,l} \bfrangle$ where $l = \max\{l_0,\dots,l_n\}$ and $m_{i,j}$ is defined as $0$ for every $j > l_i$.
\end{proof}

\begin{theorem}
	\label{thm:futs-logic-wts}
	FCL for FuTSs reduces to that for WLTSs, hence:%
	\[
		\mathfrak{L}^{\land}_{\cat{FuTS}} \reduceEq \mathfrak{L}^{\land}_{\cat{WLTS}}
		\text{.}
	\]
\end{theorem}

\begin{proof}
	Mirroring the sequence of reductions in the proof of \cref{thm:futs-to-wts}, we conclude by \cref{thm:futs-logic-u-futs,thm:futs-logic-t-futs,thm:futs-logic-h-futs,thm:ht-futs-logic-hn-futs,thm:uhn-futs-logic-wts}.
\end{proof}

As a consequence of \cref{thm:futs-logic-wts}, it is possible to invoke \cref{thm:reduction-full-abstraction} to infer from \cref{thm:wlts-logic-bisimulation} that finite-conjunction logics for FuTSs is fully abstract.

\begin{corollary}
	\label{thm:futs-logic-full-abstraction}
	For $(X,\alpha)$ a FuTS with weights drawn from a positive cancellative monoid, 
	\begin{equation}
		x \sim x' \iff x \simeq x'
		\text{.}
	\end{equation}
\end{corollary}

\begin{proof}
	The thesis follows from \cref{thm:reduction-full-abstraction,thm:futs-logic-wts,thm:wlts-logic-bisimulation}.
\end{proof}

\section{Conclusions}
\label{sec:conclusions}

In this paper we have introduced a notion of \emph{reduction} for categories of discrete state transition systems, and some general results for deriving reductions from the type of computational aspects.
As an application of this theory we have shown that FuTSs reduce to WLTSs, thus  the upper part of the hierarchy in \cref{fig:hierarchy} collapses as shown in \cref{fig:hierarchy-collapsed}.
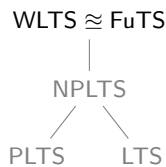
\begin{figure}[t]
	\begin{center}
		\begin{tikzpicture}[
				font=\small,auto,
				xscale=1.4, yscale=.9,
				baseline=(current bounding box.center),
			]
			\node (wlts) at (.5,2) {$\cat{WLTS}\reduceEq\cat{FuTS}$};
			\node[gray] (nplts)  at (.5,1) {\cat{NPLTS}};
			\node[gray] (lts)    at ( 1,0) {\cat{LTS}};
			\node[gray] (plts)   at ( 0,0) {\cat{PLTS}};
			
			\draw[gray] (wlts) -- (nplts);
			\draw[gray] (nplts) -- (plts);
			\draw[gray] (nplts) -- (lts);
		\end{tikzpicture}
		\caption{The bisimulation-driven hierarchy of weighted transition systems.}
		\label{fig:hierarchy-collapsed}
	\end{center}
\end{figure}
Besides the classification interest, this result offers a solid bridge for porting existing and new results from WLTSs to FuTSs. In this paper we have shown how to derive new fully abstract Hennessy-Milner modal logics for transition systems; in particular, we have introduced a new logic for FuTSs and proved that is is fully abstract via a reduction.  On this direction, SOS specifications formats presented in \cite{mp:tcs2016,ks:ic2013-sos} can cope now with FuTSs, and any abstract GSOS for these systems admits a specification in the format presented in \cite{mp:tcs2016}.

It remains an open question whether the hierarchy can be further collapsed, especially when other notion of reduction are considered. In fact, requiring a correspondence between bisimulations for the original and reduced systems may be too restrictive in some applications like bisimilarity-based verification techniques. This suggests to investigate laxer notions of reductions, such as those indicated in \cref{rm:relaxed-versions}.
Another direction is to consider different behavioural equivalences, like trace equivalence or weak bisimulation. We remark that, as shown in \cite{hjs:lmcs2007,bmp:jlamp2015,bp:concur2016}, in order to deal with these and similar equivalences, endofunctors need to be endowed with a monad (sub)structure; although WLTSs are covered in \cite{mp:arxiv2013,bmp:jlamp2015}, an analogous account of FuTSs is still an open problem.

\printbibliography

\end{document}